\documentclass[10pt,twocolumn,twoside,romanappendices]{IEEEtran}

\newcommand{\for}{\mathop{\rm for}}
\newcommand{\den}{\mathop{\rm end}}

\makeatletter
\newcommand{\vast}{\bBigg@{4}}
\newcommand{\vastt}{\bBigg@{6}}
\makeatother
\usepackage{color}
\usepackage{multirow}
\usepackage{graphicx}
\usepackage{epstopdf}
\usepackage[cmex10]{amsmath}
\usepackage{amssymb}

\usepackage{multirow}
\usepackage{algpseudocode}
\usepackage{amsthm}
\theoremstyle{definition}
\theoremstyle{remark}
\newtheorem{rem}{Remark}
\newtheorem{lem}{Lemma}
\newtheorem{thm}{Theorem}
\newtheorem{cor}{Corollary}
\newtheorem{prop}{Proposition}

\makeatletter
\g@addto@macro\th@remark{\thm@headpunct{\normalfont:}}
\makeatother

\makeatletter
\newcommand{\distas}[1]{\mathbin{\overset{#1}{\kern\z@\sim}}}%
\newsavebox{\mybox}\newsavebox{\mysim}
\newcommand{\distras}[1]{%
  \savebox{\mybox}{\hbox{\kern3pt$\scriptstyle#1$\kern3pt}}%
  \savebox{\mysim}{\hbox{$\sim$}}%
  \mathbin{\overset{#1}{\kern\z@\resizebox{\wd\mybox}{\ht\mysim}{$\sim$}}}%
	}
\allowdisplaybreaks

\begin{document}

\title{On the Joint Impact of Hardware Impairments and Imperfect CSI on Successive Decoding}

\author{Nikolaos~I.~Miridakis and Theodoros~A.~Tsiftsis,~\IEEEmembership{Senior Member,~IEEE}
\thanks{N. I. Miridakis is with the Dept. of Computer Engineering, Piraeus University of Applied Sciences, 122 44, Aegaleo, Greece (e-mail: nikozm@unipi.gr).}
\thanks{T. A. Tsiftsis is with the Dept. of Electrical Engineering, Technological Educational Institute of Central Greece, 35100 Lamia, Greece, and with the Dept. of Electrical \& Electronic Engineering, Nazarbayev University, 53 Kabanbay Batyr Ave., Astana 010000, Kazakhstan, emails:  \{tsiftsis@teiste.gr; theodoros.tsiftsis@nu.edu.kz\}.}
}

\markboth{}%
{On the Joint Impact of Hardware Impairments and Imperfect CSI on Successive Decoding}

\maketitle

\begin{abstract}
In this paper, a spatial multiplexing multiple-input multiple-output (MIMO) system when hardware along with RF imperfections occur during the communication setup is analytically investigated. More specifically, the scenario of hardware impairments at the transceiver and imperfect channel state information (CSI) at the receiver is considered, when successive interference cancellation (SIC) is implemented. Two popular linear detection schemes are analyzed, namely, zero forcing SIC (ZF-SIC) and minimum mean-square error SIC (MMSE-SIC). New analytical expressions for the outage probability of each SIC stage are provided, when independent and identically distributed Rayleigh fading channels are considered. In addition, the well-known error propagation effect between consecutive SIC stages is analyzed, while closed-form expressions are derived for some special cases of interest. Finally, useful engineering insights are manifested, such as the achievable diversity order, the performance difference between ZF- and MMSE-SIC, and the impact of imperfect CSI and/or the presence of hardware impairments to the overall system performance.
\end{abstract}


\IEEEpeerreviewmaketitle

\section{Introduction}
\IEEEPARstart{S}{patial} multiplexing represents one of the most prominent techniques used for multiple-input multiple-output (MIMO) transmission systems \cite{ref34}. In general, both linear and non-linear (e.g., maximum likelihood) detectors have been adopted in these systems. For computational savings at the receiver side, there has been a prime interest in the class of linear detectors, such as zero-forcing (ZF) and linear minimum mean-square error (MMSE). It is widely known that MMSE outperforms ZF, especially in moderately medium-to-high signal-to-noise ratio (SNR) regions, at the cost of a higher computational burden \cite{ref35}. This occurs because MMSE computes the noise variance along with the channel estimates, in contrast to ZF which processes only the channel estimates. Thereby, MMSE appropriately mitigates interference and noise, while ZF cancels interference completely but enhances the noise power at the same time. On the other hand, a simplified non-linear yet capacity-efficient method is the successive interference cancellation (SIC). It is usually combined with ZF or MMSE to appropriately counterbalance performance and computational complexity \cite{ref13}.

Performance assessment of either ZF- or MMSE-SIC has been extensively reported in the technical literature to date (e.g., see \cite{ref35}-\cite{ref36} and references therein). Nevertheless, all the previous studies assumed perfect channel state information (CSI) at the receiver and/or a non-impaired hardware at the transceiver; an ideal and a rather overoptimistic scenario for practical applications. More specifically, the hardware gear of wireless transceivers may be subject to impairments, such as I/Q imbalance, phase noise, and high power amplifier non-linearities \cite{ref37,ref3777}. These impairments are typically mitigated with the aid of certain compensation algorithms at the transceiver. Nevertheless, inadequate compensation mainly due to the imperfect parameter estimation and/or time variation of the hardware characteristics may result to residual impairments, which are added to the transmitted/received signal \cite{ref37}. Moreover, an erroneous CSI may occur due to imperfect feedback signaling and/or rapid channel variations. It can cause crosstalk interference (see \cite{ref38} and \cite{refnnneeeww} for explicit details on this effect) within the SIC process, while it can affect the detection ordering \cite{ref7}. It is noteworthy that an analytical performance assessment of ZF- and/or MMSE-SIC under the aforementioned \emph{non-ideal} communication setup (i.e., impaired hardware at the transceiver and imperfect CSI) has not been reported in the open technical literature so far.

Capitalizing on the above observations, current work presents a unified analytical performance study of ZF- and MMSE-SIC for non-ideal transmission systems. The ideal (traditional) scenario is also considered as a special case. Particularly, the ordered ZF-SIC scheme is considered, where the suboptimal yet computational efficient Foschini ordering is adopted (i.e., strongest stream is detected first, while weakest stream is detected last, upon the ZF equalization). It should be mentioned that the norm-based Foschini ordering requires only $m(m+1)/2-1$ comparisons, with $m$ denoting the number of transmit antennas. This represents a remarkable computational gain over the optimal ordering, which operates over an exhaustive search of $m!$ combinations. Interestingly, it was recently demonstrated that Foschini ordering coincides with the optimal one, in the case when the transmission rate is uniformly allocated among the transmitters \cite{ref5}. Additionally, the scenario of MMSE-SIC with a fixed-ordering is analytically presented and studied, which can serve as a benchmark for the more sophisticated ordered MMSE-SIC scheme.

The contributions of this work are summarized as follows:
\begin{itemize}
	\item New closed-form expressions for the outage probability for both the ordered ZF-SIC and unordered MMSE-SIC schemes are derived. These expressions are reduced to the corresponding conventional outage probabilities, when ideal systems are assumed with perfect CSI and without hardware impairments at the transceiver.
	\item The well-known error propagation effect between consecutive SIC steps is analyzed for the general scenario. Relevant closed-form expressions are provided for some special cases of interest.
	\item A new MMSE linear filter is presented when the variances of noise, hardware impairments and imperfect CSI are known. Based on this filter, a substantial performance gain of MMSE-SIC over ZF-SIC is observed.
	\item Simplified expressions in the asymptotically high SNR regime are obtained, revealing useful engineering insights, achievable diversity order and impact of non-ideal communication conditions to the overall system performance.
\end{itemize}

The rest of this paper is organized as follows: In Section II, the system model is presented in detail. New analytical performance results with respect to the outage probability of ZF- and MMSE-SIC are derived in Sections III and IV, respectively, while relevant asymptotic approximations are provided in Section V. The error propagation effect is analyzed in Section VI. Moreover, numerical results are presented in Section VII, while Section VIII concludes the paper.

\emph{Notation}: Vectors and matrices are represented by lowercase bold typeface and uppercase bold typeface letters, respectively. Also, $[\textbf{X}]_{ij}$, denotes the element in the \textit{i}th row and \textit{j}th column of $\mathbf{X}$, $(\mathbf{X})^{-1}$ is the inverse of $\mathbf{X}$ and $\mathbf{x}_{i}$ denotes the $i$th coefficient of $\mathbf{x}$. The superscript $(.)^{\mathcal{H}}$ denotes Hermitian transposition and $|.|$ represents absolute (scalar) value. In addition, $\mathbf{I}_{v}$ stands for the $v\times v$ identity matrix, $\mathbb{E}[.]$ is the expectation operator, $\overset{\text{d}}=$ represents equality in probability distributions, $\text{Pr}[.]$ returns probability, while $o(.)$ is the Landau symbol (i.e., $f(x)=o(g(x))$, when $f(x)/g(x)\rightarrow 0$ as $x\rightarrow \infty$). Also, $f_{X}(.)$ and $F_{X}(.)$ represent probability density function (PDF) and cumulative distribution function (CDF) of the random variable (RV) $X$, respectively. Complex-valued Gaussian RVs with mean $\mu$ and variance $\sigma^{2}$, while chi-squared RVs with $v$ degrees-of-freedom are denoted, respectively, as $\mathcal{CN}(\mu,\sigma^{2})$ and $\mathcal{X}^{2}_{2v}$. Furthermore, $\Gamma(a)\triangleq (a-1)!$ (with $a\in \mathbb{N}^{+}$) denotes the Gamma function \cite[Eq. (8.310.1)]{ref1}, $B(a,b)\triangleq \Gamma(a)\Gamma(b)/\Gamma(a+b)$ is the Beta function \cite[Eq. (8.384.1)]{ref1} and $\mathcal{U}(a,b,x)\triangleq \int^{\infty}_{0}\exp(-xt)t^{a-1}(t+1)^{b-a-1}/\Gamma(a)dt$ (with $\{a,x\}>0$) corresponds to the Tricomi confluent hypergeometric function \cite[Eq. (9.211.4)]{ref1}.

\section{System Model}
Consider a point-to-point MIMO system where the transmitter and receiver sides are equipped with $m$ and $n\geq m$ antennas, respectively. The input-output relation of the received signal stems as \cite{ref20}
\begin{equation}
\mathbf{y}=\mathbf{H}\left(\textbf{s}+\textbf{n}_{T}\right)+\textbf{n}_{R}+\textbf{w},
\label{inouttt}
\end{equation}
where $\textbf{y} \in \mathbb{C}^{n \times 1}$, $\textbf{s} \in \mathbb{C}^{m \times 1}$ and $\textbf{w} \in \mathbb{C}^{n \times 1}$ denote the received, the transmit and the circularly symmetric Gaussian noise signal vectors, respectively. In addition, $\textbf{n}_{T} \in \mathbb{C}^{m \times 1}$ and $\textbf{n}_{R} \in \mathbb{C}^{n \times 1}$ correspond to the distortion noise due to residual hardware impairments at the transmitter and receiver, respectively.\footnote{This distortion noise denotes the \emph{aggregation} of many residual impairments when compensation algorithms are applied to mitigate the main hardware impairments \cite{newone}.} Moreover, $\mathbf{H} \in \mathbb{C}^{n \times m}$ is the channel matrix, while assuming that the coefficients of $\textbf{H}\overset{\text{d}}= \mathcal{CN}(0,1)$, i.e., a Rayleigh flat fading scenario. Also, $\mathbb{E}[\textbf{ww}^{\mathcal{H}}]=N_{0}\textbf{I}_{n}$, where $N_{0}$ is the noise power, while $\mathbb{E}[\textbf{ss}^{\mathcal{H}}]=p\textbf{I}_{m}$ is assumed, where $p$ denotes the transmitted power per antenna. Typically, $\textbf{n}_{T}$ and $\textbf{n}_{R}$ are Gaussian distributed (see, e.g., \cite{ref20} and references therein), i.e., $\textbf{n}_{T}\overset{\text{d}}= \mathcal{CN}(0,p\kappa^{2}_{T}\textbf{I}_{m})$ and $\textbf{n}_{R}\overset{\text{d}}= \mathcal{CN}(0,p\kappa^{2}_{R}m\textbf{I}_{n})$, where $\kappa_{T}$ and $\kappa_{R}$ denote the level of residual impairments\footnote{In practical systems, $\kappa_{T}$ is equivalent to the error vector magnitude \cite{ref21}, which is defined as the ratio of distortion-to-signal magnitude, and can be measured directly with the aid of \cite{refevm}. As an indicative example, typical values of $\kappa_{T}$ in LTE infrastructures \cite{ref21} lie in the range of [$0.08, 0.175$].} at the transmitter and receiver, respectively. It is noteworthy that the variance of residual impairments is proportional to the transmission power per antenna \cite[Eqs. (7) and (8)]{ref20}. Also, the last two terms of (\ref{inouttt}), i.e., $\textbf{n}_{R}+\textbf{w}$, denote the total post-noise added onto the received signal, which can be modeled as $\mathcal{CN}(0,(p\kappa^{2}_{R}m+N_{0})\textbf{I}_{n})$.

 In the ideal scenario where $\left\{\kappa_{T},\kappa_{R}\right\}=0$ (i.e., no hardware impairments), (\ref{inouttt}) is reduced to the conventional MIMO signal relation, given by
\begin{align*}
\mathbf{y}=\mathbf{H}\textbf{s}+\textbf{w}.
\end{align*}

Further, in the rather realistic scenario when imperfect CSI occurs, the estimated channel at the receiver is given by
\begin{equation}
\mathbf{\hat{H}}\triangleq \mathbf{H}+\mathbf{\Delta H},
\label{inout}
\end{equation}
where $\mathbf{\hat{H}}$ is the estimated channel matrix, $\mathbf{\Delta H} \in \mathbb{C}^{n \times m}$ stands for the channel estimation error matrix, while the coefficients of $\mathbf{\Delta H}\overset{\text{d}}= \mathcal{CN}(0,\omega)$ with $\omega$ representing the channel estimation error variance \cite{ref7}. Also, $\mathbf{H}$ and $\mathbf{\Delta H}$ are statistically independent \cite{ref29}.

In the following, we turn our focus on two quite popular linear detection schemes, namely, ZF and MMSE. These schemes, combined with SIC, are extensively used in spatial multiplexing transmissions \cite{ref13}.

\subsection{ZF-SIC}
In principle, ZF-SIC enables spatial multiplexing transmission, i.e., it can distinguish the received streams from different users and/or antennas with the aid of spatial structures (individual spatial signatures) of the signals to be detected \cite{ref2}. It is performed in three main steps, namely, the \emph{symbol ordering} that aims to enhance the overall reception performance, the \emph{interference nulling} via ZF from the yet-to-be detected symbols, and the \emph{interference cancellation} from the already detected symbols. These steps are performed in a number of consecutive stages, until all given symbols are successfully decoded.

The interference nulling can be efficiently implemented by applying the QR decomposition on a given channel matrix, which is widely adopted in ZF equalizers, since it provides computational complexity savings \cite{ref3}. Let $\mathbf{\hat{Q}}$ be a $n\times n$ unitary matrix (with its columns representing the orthonormal ZF nulling vectors) and $\mathbf{\hat{R}}$ an $n\times m$ upper triangular matrix, given $\mathbf{\hat{H}}$. Accordingly, $\mathbf{Q}$ and $\mathbf{R}$ correspond to the true channel matrix $\mathbf{H}$. It follows from (\ref{inout}) that 
\begin{align}
\nonumber
\mathbf{\hat{Q}}\mathbf{\hat{R}}&= \mathbf{Q}\mathbf{R}+\mathbf{\Delta H}\\
\Leftrightarrow\mathbf{\hat{Q}}^{\mathcal{H}}&=(\mathbf{Q}\mathbf{R}\mathbf{\hat{R}}^{-1})^{\mathcal{H}}+(\mathbf{\hat{R}}^{-1})^{\mathcal{H}}\mathbf{\Delta H}^{\mathcal{H}}.
\label{inouttttt}
\end{align}
Hence, $\mathbf{\hat{Q}}^{\mathcal{H}}\mathbf{y}$ is performed at the receiver, yielding
\begin{align}
\nonumber
\mathbf{\hat{Q}}^{\mathcal{H}}\mathbf{y}&=\mathbf{\hat{Q}}^{\mathcal{H}}\left(\mathbf{Q}\mathbf{R}\left(\textbf{s}+\textbf{n}_{T}\right)+\textbf{n}_{R}+\textbf{w}\right)\\
\nonumber
&=\left((\mathbf{Q}\mathbf{R}\mathbf{\hat{R}}^{-1})^{\mathcal{H}}+(\mathbf{\hat{R}}^{-1})^{\mathcal{H}}\mathbf{\Delta H}^{\mathcal{H}}\right)\mathbf{Q}\mathbf{R}\left(\textbf{s}+\textbf{n}_{T}\right)\\
&\ \ \ +\mathbf{\hat{Q}}^{\mathcal{H}}\textbf{n}_{R}+\mathbf{\hat{Q}}^{\mathcal{H}}\textbf{w}.
\label{referr}
\end{align}
Interestingly, it has been demonstrated in \cite[Eq. (30)]{ref7} and \cite[Eq. (16)]{ref38} that $\mathbf{\hat{R}}\approx \mathbf{R}$, whereas the resultant approximation error can be considered negligible in terms of distributions \cite{ref7}. Also, note that the latter approximations become exact equalities in the case when perfect CSI is available. Thereby, (\ref{referr}) can be reformed as
\begin{align}
\nonumber
&\mathbf{\hat{Q}}^{\mathcal{H}}\mathbf{y}\approx\\
&\left(\mathbf{I}_{n}+(\mathbf{R}^{-1})^{\mathcal{H}}\mathbf{\Delta H}^{\mathcal{H}}\mathbf{Q}\right)\mathbf{R}\left(\textbf{s}+\textbf{n}_{T}\right)+\mathbf{\hat{Q}}^{\mathcal{H}}\textbf{n}_{R}+\mathbf{\hat{Q}}^{\mathcal{H}}\textbf{w}.
\label{refer}
\end{align}
Thus, the sequential signal decoding, which involves the decision feedback, is given by
\begin{flalign*}
&\for\:\:i=m:-1:1 &\\
&\ \ \ \ \:\hat{\mathbf{s}}_{i}=\mathcal{Q}\left[\frac{\left(\mathbf{\hat{Q}}^{\mathcal{H}}\mathbf{y}\right)_{i}-\sum^{m}_{j=i+1}\hat{r}_{ij}\hat{s}_{j}}{\hat{r}_{ii}}\right]& \\
&\ \ \ \ \ \ \:\approx \mathcal{Q}\left[\frac{\left(\mathbf{\hat{Q}}^{\mathcal{H}}\mathbf{y}\right)_{i}-\sum^{m}_{j=i+1}r_{ij}\hat{s}_{j}}{r_{ii}}\right]& \\
&\den&
\end{flalign*}
where $\hat{\mathbf{s}}_{i}$ is the estimated symbol of the $i$th detected stream, $\hat{r}_{ij}$ (or $r_{ij}$) is the coefficient at the $i$th row and $j$th column of $\mathbf{\hat{R}}$ (or $\mathbf{R}$) and $\mathcal{Q}[.]$ stands for the slicing operator mapping to the nearest point in the symbol constellation.

Therefore, based on the unitary invariant property of Gaussian vectors (i.e., isotropic distribution \cite[Theorem 1.5.5]{ref4}), the signal-to-interference-plus-noise-and-distortion ratio (SINDR) of the $i$th decoding layer\footnote{The forward decoding is adopted into this work and, therefore, the first SIC stage corresponds to the last decoding layer of the processing matrix (from the left to the right). Similarly, the \textit{i}th decoding layer corresponds to the ($m-i+1$)th SIC stage. Note that the terms \textit{decoding layer} and \textit{SIC stage} will be interchangeably used in the rest of this paper.} for ZF-SIC is expressed as
\begin{align}
\nonumber
&\text{SINDR}_{i}\\
&\approx \textstyle \frac{p r^{2}_{ii}}{pr^{2}_{ii}\kappa^{2}_{T}+p\sum^{m}_{j=1}\left|\left((\mathbf{R}^{-1})^{\mathcal{H}}\mathbf{\Delta H}^{\mathcal{H}}\mathbf{Q}\mathbf{R}\right)_{ij}\right|^{2}(1+\kappa^{2}_{T})+p\kappa^{2}_{R}m+N_{0}}.
\label{sindr}
\end{align}
Notice that in the ideal scenario of perfect CSI and no hardware impairments, (\ref{sindr}) becomes the classical SNR expression of the $i$th layer, since $\text{SNR}_{i}=pr^{2}_{ii}/N_{0}$ \cite{ref33}.

\subsection{MMSE-SIC}
Unlike ZF-SIC, the more sophisticated MMSE-SIC detector achieves an optimal balance between interference suppression and noise enhancement. To this end, it requires the knowledge (or estimation) of the noise variance and, thus, it represents the optimal linear detection scheme \cite[App. A]{ref23}. Since the main difference between ZF- and MMSE-SIC is in the equalization process, we retain our focus on the discussion of the typical MMSE, while the description of the more advanced MMSE-SIC is provided subsequently.

The conventional MMSE (non-SIC) detector strives to minimize the mean-square error (MSE) of the $j$th transmitted stream, i.e., $s^{(j)}$, as follows
\begin{align}
\text{MSE}^{(j)}=\mathbb{E}\left[\left|s^{(j)}-(\mathbf{g}^{(j)})^{\mathcal{H}}\mathbf{\hat{y}}\right|^{2}\right],\ \ 1\leq j\leq m,
\label{mse}
\end{align}
where $\mathbf{g}^{(j)}$ is the optimal weight vector and $\mathbf{\hat{y}}$ denotes the post-detection received signal, subject to channel estimation imperfections and hardware impairments of the transceiver. To facilitate the analysis, we can formulate $\mathbf{\hat{y}}$ as the classical MIMO model
\begin{align}
\mathbf{y}=\mathbf{H}\mathbf{s}+\mathbf{w}',
\end{align}
where $\mathbf{w}'\triangleq \left(\mathbf{H}+\mathbf{\Delta H}\right)\mathbf{n}_{T}+\mathbf{\Delta H}\mathbf{s}+\mathbf{n}_{R}+\mathbf{w}$ with a (colored) noise covariance matrix given by \cite[Eq. (9)]{ref20}
\begin{align}
\nonumber
\mathbb{E}[\mathbf{w}'\mathbf{w}'^{\mathcal{H}}]&=p\kappa^{2}_{T}\left(\mathbf{H}+\mathbf{\Delta H}\right)\left(\mathbf{H}+\mathbf{\Delta H}\right)^{\mathcal{H}}\\
&+p\mathbf{\Delta H}\left(\mathbf{\Delta H}\right)^{\mathcal{H}}+(p\kappa^{2}_{R}m+N_{0})\mathbf{I}_{n}.
\end{align}
Due to the scaling property of Gaussian RVs \cite[Chapt. 3]{refscaleprop}, while keeping in mind the independence between $\mathbf{H}$ and $\mathbf{\Delta H}$, it holds that $\mathbb{E}[(\mathbf{\Delta H})(\mathbf{\Delta H})^{\mathcal{H}}]=\omega \mathbb{E}[\mathbf{H}\mathbf{H}^{\mathcal{H}}]$. Hence, after some simple manipulations, the noise covariance matrix can be expressed more concisely as
\begin{align}
\mathbb{E}[\mathbf{w}'\mathbf{w}'^{\mathcal{H}}]=(p\kappa^{2}_{T}(\omega+1)+p\omega)\mathbf{H}\mathbf{H}^{\mathcal{H}}+(p\kappa^{2}_{R}m+N_{0})\mathbf{I}_{n}.
\end{align}
Based on (\ref{mse}), it can be seen that (see Appendix \ref{app0} for details)
\begin{align}
\nonumber
&\mathbf{g}^{(j)}=p\left(p\mathbf{H}\mathbf{H}^{\mathcal{H}}+\mathbb{E}[\mathbf{w}'\mathbf{w}'^{\mathcal{H}}]\right)^{-1}\mathbf{h}_{j}\\
&=\textstyle \left(\mathbf{H}\mathbf{H}^{\mathcal{H}}\left(\scriptstyle \kappa^{2}_{T}(\omega+1)+\omega+1\right)+\left(\scriptstyle \kappa^{2}_{R}m+\frac{N_{0}}{p}\right)\mathbf{I}_{n}\right)^{-1}\mathbf{h}_{j},
\label{filterg}
\end{align}
whereas, after some straightforward manipulations (see Appendix \ref{app0}), the total SINDR of the $j$th stream is obtained as \cite[Eq. (5)]{ref24}
\begin{align}
\nonumber
&\text{SINDR}^{(j)}=\\
&\frac{\frac{1}{(\kappa^{2}_{R}m+N_{0}/p)}\mathbf{h}_{j}^{\mathcal{H}}\left(\mathbf{H}\mathbf{H}^{\mathcal{H}}\frac{(\kappa^{2}_{T}(\omega+1)+\omega+1)}{(\kappa^{2}_{R}m+N_{0}/p)}+\mathbf{I}_{n}\right)^{-1}\mathbf{h}_{j}}{1-\frac{(2\sqrt{\omega}+1)^{-1}}{(\kappa^{2}_{R}m+N_{0}/p)}\mathbf{h}_{j}^{\mathcal{H}}\left(\mathbf{H}\mathbf{H}^{\mathcal{H}}\frac{(\kappa^{2}_{T}(\omega+1)+\omega+1)}{(\kappa^{2}_{R}m+N_{0}/p)}+\mathbf{I}_{n}\right)^{-1}\mathbf{h}_{j}}.
\label{sindrmmse}
\end{align}
Based on Woodbury's identity \cite[Eq. (2.1.4)]{ref28}, (\ref{sindrmmse}) reads also as
\begin{align}
\text{SINDR}^{(j)}=\frac{\mathcal{C}^{(j)}}{1-\frac{\mathcal{C}^{(j)}}{2\sqrt{\omega}+1}},
\label{sindrmmse1}
\end{align}
where
\begin{align*}
\nonumber
\mathcal{C}^{(j)}&\triangleq \frac{1}{(\kappa^{2}_{T}(\omega+1)+\omega+1)}\\
&\times \frac{\mathbf{h}_{j}^{\mathcal{H}}\left(\mathbf{K}_{j}\mathbf{K}_{j}^{\mathcal{H}}+\frac{(\kappa^{2}_{R}m+N_{0}/p)}{(\kappa^{2}_{T}(\omega+1)+\omega+1)}\mathbf{I}_{n}\right)^{-1}\mathbf{h}_{j}}{1+\mathbf{h}_{j}^{\mathcal{H}}\left(\mathbf{K}_{j}\mathbf{K}_{j}^{\mathcal{H}}+\frac{(\kappa^{2}_{R}m+N_{0}/p)}{(\kappa^{2}_{T}(\omega+1)+\omega+1)}\mathbf{I}_{n}\right)^{-1}\mathbf{h}_{j}},
\end{align*}
and $\mathbf{K}_{j}\triangleq [\mathbf{h}_{1} \cdots \mathbf{h}_{j-1}\:\: \mathbf{h}_{j+1}\cdots \mathbf{h}_{m}]$. The form of (\ref{sindrmmse1}) is preferable than (\ref{sindrmmse}) for further analysis, because $\mathbf{h}_{j}$ and $\mathbf{K}_{j}$ are statistically independent. Also, in ideal conditions of perfect CSI and no hardware impairments, (\ref{sindrmmse1}) is reduced to the classical signal-to-interference-plus-noise ratio (SINR) expression of MMSE detectors \cite[Eqs. (11) and (13)]{ref14}
\begin{align}
\text{SINR}^{(j)}=\mathbf{h}_{j}^{\mathcal{H}}\left(\mathbf{K}_{j}\mathbf{K}_{j}^{\mathcal{H}}+\frac{N_{0}}{p}\mathbf{I}_{n}\right)^{-1}\mathbf{h}_{j}.
\end{align}

On the other hand, when MMSE-SIC is applied to the receiver, the corresponding SINDR of the $i$th SIC step ($1\leq i< m$) can be expressed as
\begin{align}
\text{SINDR}_{i}=\frac{\hat{\mathcal{C}}^{(i)}}{1-\frac{\hat{\mathcal{C}}^{(i)}}{2\sqrt{\omega}+1}},
\label{sindrmmsesic}
\end{align}
where $\hat{\mathcal{C}}^{(i)}$ is the same as $\mathcal{C}^{(i)}$, but replacing $\mathbf{K}_{i}$ with $\hat{\mathbf{K}}_{i} \in \mathbb{C}^{n \times (m-i)}$, which is the remaining (deflated) version of $\mathbf{K}_{i}$ with its ($i-1$) columns being removed. This occurs because MMSE-SIC at the $i$th SIC stage is equivalent to the classical MMSE detector with the previous ($i-1$) symbols already detected.

Further, in the last SIC stage where $i=m$, it can be seen that (see Appendix \ref{app0})
\begin{align}
\nonumber
\text{SINDR}_{m}&=\frac{1}{(\kappa^{2}_{R}m+N_{0}/p)}\\
&\times \mathbf{h}_{m}^{\mathcal{H}}\left(\mathbf{h}_{m}\mathbf{h}_{m}^{\mathcal{H}}\frac{(\kappa^{2}_{T}(\omega+1)+\omega)}{(\kappa^{2}_{R}m+N_{0}/p)}+\mathbf{I}_{n}\right)^{-1}\mathbf{h}_{m},
\label{sindrmmsesicm}
\end{align}
since no inter-stream interference is experienced at the last SIC stage.\footnote{In fact, (\ref{sindrmmsesicm}) represents the optimal combining scheme in interference-free environments. In other words, it coincides with the maximal ratio combining (MRC) scheme, when imperfect CSI and hardware-impaired transceivers are present. Notice that when $\{\omega,\kappa_{T},\kappa_{R}\}=0$, (\ref{sindrmmsesicm}) is reduced to the classical SNR expression of MRC.}

\section{Performance Analysis of the Ordered ZF-SIC}
In this section, closed-form formulae with regards to the outage performance of the ordered ZF-SIC for each transmitted stream are provided. We start from the general scenario, when both CSI errors and hardware impairments are present, followed by some simplified special cases of interest.

\subsection{General Case}
We commence by deriving the CDF of the SINDR for each transmitted stream, which represents the corresponding outage probability, as follows.
\begin{align}
\nonumber
&\text{Pr}\left[\text{SINDR}_{i}\leq \gamma_{\text{th}}\right]\Leftrightarrow\\
&\text{Pr}\left[p r^{2}_{ii}\leq \frac{\left(p\left(\kappa^{2}_{T}+1\right)Y_{i}+p\kappa^{2}_{R}m+N_{0}\right)\gamma_{\text{th}}}{\left(1-\kappa^{2}_{T}\gamma_{\text{th}}\right)}\right],
\label{cdf}
\end{align}
where $\gamma_{\text{th}}$ denotes the predetermined SINDR outage threshold, while the auxiliary variable $Y_{i}\triangleq \sum^{m}_{j=1}|((\mathbf{R}^{-1})^{\mathcal{H}}\mathbf{\Delta H}^{\mathcal{H}}\mathbf{Q}\mathbf{R})_{ij}|^{2}$ is introduced for notational convenience. Notice that the condition $\kappa^{2}_{T} < 1/\gamma_{\text{th}}$ should be satisfied, which is typically the case in most practical applications. Thus, it holds that
\begin{align}
\nonumber
&P^{(i)}_{\text{out}}(\gamma_{\text{th}})\triangleq F_{\text{SINDR}_{i}}(\gamma_{\text{th}})\approx\\
&1-\text{Pr}\left[p r^{2}_{ii}\geq \frac{\left(p\left(\kappa^{2}_{T}+1\right)Y_{i}+p\kappa^{2}_{R}m+N_{0}\right)\gamma_{\text{th}}}{\left(1-\kappa^{2}_{T}\gamma_{\text{th}}\right)}\right],
\label{cdf1}
\end{align}
where $P^{(i)}_{\text{out}}(.)$ denotes the outage probability for the $i$th stream.

To proceed, we have to determine the distributions of the mutually independent RVs, namely, $Y_{i}$ and $p r^{2}_{ii}$.

\begin{lem}
The PDF of $Y_{i}$, $f_{Y_{i}}(.)$, yields as
\begin{align}
f_{Y_{i}}(x)=\frac{x^{m-1}\exp \left(-\frac{x}{\omega}\right)}{\Gamma(m)\omega^{m}},\ \forall i,\ 1\leq i\leq m.
\label{pdferror}
\end{align}
\end{lem}

\begin{proof}
From \cite{ref7}, while conditioning on $\mathbf{R}$, in a similar manner as in \cite[Eq. (11)]{ref16neww}, $Y_{i}\overset{\text{d}}=\frac{\omega}{2}\mathcal{G}_{i}$, where $\mathcal{G}_{i}\overset{\text{d}}=\mathcal{X}^{2}_{2m}$. Based on the scaling property of RVs (i.e., $f_{Z=c X}(z)=f_{X}(z/c)/c$ for $c\geq 0$), the result in (\ref{pdferror}) is obtained.
\end{proof}

On the other hand, $f_{p r^{2}_{ii}}(.)$ depends on the precise ordering that is adopted. In current study, the classical Foschini (norm-based) ordering is investigated, where the strongest stream is decoded first while the weakest stream is decoded last. It was recently demonstrated that the Foschini ordering coincides with the optimal ordering in the case when the transmission rate is uniformly allocated among the transmitters \cite{ref5}.

\begin{lem}
In the case when Foschini ordering is applied, $f_{p r^{2}_{ii}}(.)$ is given by
\begin{equation}
f_{p r^{2}_{ii}}(x)=\Xi_{i}\:x^{\xi_{i}}\exp \left(-\frac{(m+l-i+1)x}{p}\right),
\label{pdfrii}
\end{equation}
where
\begin{align}
\nonumber
&\Xi_{i}\triangleq \sum^{i-2}_{j=0}\sum^{i-1}_{l=0}\sum^{m+l-i}_{\rho_{1}=0}\sum^{\rho_{1}}_{\rho_{2}=0}\cdots\sum^{\rho_{n-2}}_{\rho_{n-1}=0}\sum^{i+\phi-j-2}_{r=0}(i+\phi-j-2)!\\
\nonumber
&\times \prod^{n-1}_{t=1}\left[\frac{(-1)^{j+l}\binom{i-2}{j}\binom{i-1}{l}}{(\rho_{t-1}-\rho_{t})!(t!)^{\rho_{t}-\rho_{t+1}}}\right]\frac{p^{i-j-r-n-1}}{r!\rho_{n-1}!(n-1)!}\\
&\times \frac{(m+l-i)!(m+l-i+1)^{-(i+\phi-j-r-1)}}{B(n-i+1,i-1)B(m-i+1,i)},\ \ i>1,
\label{Xi1}
\end{align}
or
\begin{align}
\nonumber
\Xi_{1}&\triangleq \sum^{m-1}_{\rho_{1}=0}\sum^{\rho_{1}}_{\rho_{2}=0}\cdots\sum^{\rho_{n-2}}_{\rho_{n-1}=0}\frac{m!}{\rho_{n-1}!p^{n+\phi}(n-1)!}\\
&\times \prod^{n-1}_{t=1}\left[\frac{1}{(\rho_{t-1}-\rho_{t})!(t!)^{\rho_{t}-\rho_{t+1}}}\right],\ \ i=1,
\label{Xi2}
\end{align}
while $\xi_{i}\triangleq n+r+j-i$ (for $i>1$), $\xi_{1}\triangleq n+\phi-1$, $\rho_{0}\triangleq m+l-i$ for $i>1$ or $\rho_{0}\triangleq j$ for $i=1$, $\rho_{n}\triangleq 0$, and $\phi\triangleq \sum^{n-1}_{q=1}\rho_{q}$. In general, $m+l-i$ is substituted with $j$ in the case of $i=1$.
\end{lem}

\begin{proof}
The detailed proof is relegated in \cite{ref6}.
\end{proof}

In the simplified scenario of fixed symbol ordering (i.e., no ordering), $f_{p r^{2}_{ii}}(.)\overset{\text{d}}=\mathcal{X}^{2}_{2(n-i+1)}$ \cite{ref8}. Thereby, in this case, $\Xi_{i}\triangleq 1/(\Gamma(n-i+1)p^{n-i+1})$ and $\xi_{i}\triangleq n-i$ for $1\leq i\leq m$, while $\exp (-(m+l-i+1)x/p)$ in (\ref{pdfrii}) is replaced with $\exp (-x/p)$.

We are now in a position to formulate the outage probability for the ordered ZF-SIC as follows:
\begin{thm}
Outage probability for the $i$th decoding layer is obtained in closed-form as
\begin{align}
\nonumber
P^{(i)}_{\text{out}}(\gamma_{\text{th}})&\approx 1-\Psi_{i} \sum^{\mu}_{v=0}\frac{\binom{\mu}{v}\left(p\kappa^{2}_{R}m+N_{0}\right)^{\mu-v}\left(p\left(\kappa^{2}_{T}+1\right)\right)^{v}}{\Gamma(m)\omega^{m}(1-\kappa^{2}_{T}\gamma_{\text{th}})^{\mu}}\\
&\times \frac{\gamma_{\text{th}}^{\mu}\:\Gamma(v+m)\exp \left(-\frac{(m+l-i+1)\gamma_{\text{th}}(p\kappa^{2}_{R}m+N_{0})}{p(1-\kappa^{2}_{T}\gamma_{\text{th}})}\right)}{\left(\frac{(m+l-i+1)\gamma_{\text{th}}(\kappa^{2}_{T}+1)}{(1-\gamma_{\text{th}}\kappa^{2}_{T})}+\frac{1}{\omega}\right)^{v+m}},
\label{outclosed}
\end{align}
where
\begin{align*}
\Psi_{i}\triangleq \Xi_{i} \sum^{\xi_{i}}_{\mu=0}\frac{\xi_{i}!}{\mu!\left(\frac{m+l-i+1}{p}\right)^{\xi_{i}-\mu+1}}.
\end{align*}
\end{thm}

\begin{proof}
The proof is provided in Appendix \ref{appa}.
\end{proof}

It is noteworthy that the derived result includes finite sum series of simple elementary functions and factorials and, thus, can be efficiently and rapidly calculated.\footnote{At this point, it should be mentioned that the auxiliary parameters $\Xi_{i}$ and $\Psi_{i}$ include the required multiple nested sum series, while they are introduced for notational simplicity and presentation compactness.}

\subsection{Imperfect CSI without hardware impairments}
In this case, the system suffers from imperfect CSI, which in turn reflects to channel estimation errors, but it is equipped with ideal hardware. The corresponding outage probability of each stream is directly obtained from (\ref{outclosed}), by setting $\kappa_{T}=\kappa_{R}=0$.

\subsection{Perfect CSI with hardware impairments}
This scenario corresponds to the case when channel is correctly estimated (e.g., via pilot or feedback signaling), but the transmitted and/or received signal is impaired due to low-cost hardware equipment at the transceivers.
\begin{prop}
The exact closed-form outage probability of the $i$th stream under perfect CSI conditions with hardware impairments is expressed as
\begin{align}
\nonumber
P^{(i)}_{\text{out}}(\gamma_{\text{th}})=&1-\Psi_{i} \left(\frac{\left(p\kappa^{2}_{R}m+N_{0}\right)\gamma_{\text{th}}}{\left(1-\kappa^{2}_{T}\gamma_{\text{th}}\right)}\right)^{\mu}\\
&\times \exp\left(-\frac{(m+l-i+1)\left(p\kappa^{2}_{R}m+N_{0}\right)\gamma_{\text{th}}}{p\left(1-\kappa^{2}_{T}\gamma_{\text{th}}\right)}\right).
\label{outclosed1}
\end{align}
\end{prop}

\begin{proof}
The proof is given in Appendix \ref{appb}.
\end{proof}

\section{Performance Analysis of MMSE-SIC with Fixed Ordering}
A closed-form expression for the PDF/CDF of SINDR with regards to the ordered MMSE-SIC is not available so far. To this end, we retain our focus on the unordered (fixed) MMSE-SIC scenario in this section, which can be used as a benchmark and/or as a lower performance bound for the more sophisticated ordered MMSE-SIC scheme.

\begin{thm}
Outage probability of the $i$th SIC stage, when $1\leq i < m$, is derived in a closed-form as given by (\ref{cdfsindrmmsesic})
\begin{figure*}
\begin{align}
\nonumber
&P^{(i)}_{\text{out}}(\gamma_{\text{th}})=\\
\nonumber
&1-\exp\left(-\frac{\left(\kappa^{2}_{R}m+\frac{N_{0}}{p}\right)\gamma_{\text{th}}}{1+\gamma_{\text{th}}\left(1-\left(\kappa^{2}_{T}(\omega+1)+\omega+1\right)\right)}\right)\vastt[\sum^{n}_{k_{1}=1}\frac{1}{(k_{1}-1)!}\left(\frac{\left(\kappa^{2}_{R}m+\frac{N_{0}}{p}\right)\gamma_{\text{th}}}{1+\gamma_{\text{th}}\left(1-\left(\kappa^{2}_{T}(\omega+1)+\omega+1\right)\right)}\right)^{k_{1}-1}\\
&-\sum^{n}_{k_{2}=n-m+i+1}\:\:\sum^{m-i}_{j=n-k_{2}+1}\frac{\binom{m-i}{j}\left(\frac{\left(\kappa^{2}_{R}m+\frac{N_{0}}{p}\right)}{\left(\kappa^{2}_{T}(\omega+1)+\omega+1\right)}\right)^{k_{2}-1}\left(\frac{\gamma_{\text{th}}}{\gamma_{\text{th}}\left(\frac{(2\sqrt{\omega}+1)^{-1}}{\left(\kappa^{2}_{T}(\omega+1)+\omega+1\right)}-1\right)+\frac{1}{\left(\kappa^{2}_{T}(\omega+1)+\omega+1\right)}}\right)^{k_{2}+j-1}}{(k_{2}-1)!\left(1+\left(\frac{\gamma_{\text{th}}}{\gamma_{\text{th}}\left(\frac{(2\sqrt{\omega}+1)^{-1}}{\left(\kappa^{2}_{T}(\omega+1)+\omega+1\right)}-1\right)+\frac{1}{\left(\kappa^{2}_{T}(\omega+1)+\omega+1\right)}}\right)\right)^{m-i}}\vastt]
\label{cdfsindrmmsesic}
\end{align}
\hrulefill
\end{figure*}
and for the $m$th SIC stage as
\begin{align}
\nonumber
P^{(m)}_{\text{out}}(\gamma_{\text{th}})&=1-\exp\left(-\frac{\left(\kappa^{2}_{R}m+\frac{N_{0}}{p}\right)\gamma_{\text{th}}}{\left(1-\left(\kappa^{2}_{T}(\omega+1)+\omega\right)\gamma_{\text{th}}\right)}\right)\\
&\times \sum^{n-1}_{k=0}\frac{\left(\frac{\left(\kappa^{2}_{R}m+\frac{N_{0}}{p}\right)\gamma_{\text{th}}}{\left(1-\left(\kappa^{2}_{T}(\omega+1)+\omega\right)\gamma_{\text{th}}\right)}\right)^{k}}{k!}.
\label{cdfsindrmmsesicm}
\end{align}
\end{thm}

\begin{proof}
The proof is provided in Appendix \ref{appd}.
\end{proof}

In general, $\gamma_{\text{th}}<\frac{1}{(\kappa^{2}_{T}(\omega+1)+\omega)}$ should hold for the evaluation of every SIC stage (see Appendix \ref{appd} for details). When the latter condition is not satisfied, an outage occurs with probability one. As previously mentioned, typically $\kappa_{T}\leq 0.175$ \cite{ref21}. Moreover, practical values of $\omega$ could not exceed $30\%$ (i.e., $\omega\leq 0.3$), because higher values of channel estimation error reflect to a rather catastrophic reception \cite{ref7,ref29}. Thereby, based on the latter extreme values, $\gamma_{\text{th}}< 4.27$dB is required. Equivalently, since $\gamma_{\text{th}}\triangleq 2^{\mathcal{R}}-1$ (where $\mathcal{R}$ denotes a target transmission rate), $\mathcal{R}<1.88$bps/Hz is required for a feasible communication. Nonetheless, higher $\gamma_{\text{th}}$ values can be admitted for more relaxed CSI imperfections and/or hardware impairments, while there is no constraint in the ideal scenario.

Moreover, notice that the special cases of non-impaired hardware or perfect CSI are directly obtained by setting $\kappa_{T}=\kappa_{R}=0$ or $\omega=0$ in (\ref{cdfsindrmmsesic}) and (\ref{cdfsindrmmsesicm}), respectively.

\begin{cor}
The ideal scenario of non-impaired hardware at the transceiver and perfect CSI conditions corresponds to the typical MMSE-SIC outage probability for the $i$th SIC stage (when $1\leq i<m$), given by
\begin{align}
\nonumber
&P^{(i)}_{\text{out}}(\gamma_{\text{th}})=1-\exp\left(-\frac{N_{0}\gamma_{\text{th}}}{p}\right)\Bigg[\sum^{n}_{k_{1}=1}\frac{\left(\frac{N_{0}\gamma_{\text{th}}}{p}\right)^{k_{1}-1}}{(k_{1}-1)!}\\
&-\sum^{n}_{k_{2}=n-m+i+1}\:\:\sum^{m-i}_{j=n-k_{2}+1}\frac{\binom{m-i}{j}\left(\frac{N_{0}}{p}\right)^{k_{2}-1}\gamma_{\text{th}}^{k_{2}+j-1}}{(k_{2}-1)!\left(1+\gamma_{\text{th}}\right)^{m-i}}\Bigg],
\label{cdfsindrmmsesic11}
\end{align}
while for the $m$th SIC stage is expressed as
\begin{align}
P^{(m)}_{\text{out}}(\gamma_{\text{th}})=1-\exp\left(-\frac{N_{0}\gamma_{\text{th}}}{p}\right)\sum^{n-1}_{k=0}\frac{\left(\frac{N_{0}\gamma_{\text{th}}}{p}\right)^{k}}{k!},
\label{cdfsindrmmsesicm111}
\end{align}
which coincides with the outage probability of the conventional MRC, as it should be.
\end{cor}

\section{Asymptotic Analysis}
Although the previous formulae are presented in closed formulations, it is rather difficult to reveal useful insights, straightforwardly. Therefore, in this section, outage probability is analyzed in the asymptotically high SINDR regime. Thus, more amenable expressions are manifested, while important outcomes regarding the influence of imperfect CSI and hardware impairments are obtained.

\subsection{Ordered ZF-SIC}
\subsubsection{General Case}
The following proposition presents a sharp outage floor for the general scenario of erroneous CSI under hardware impairments.
\begin{prop}
When $\frac{p}{N_{0}}\rightarrow \infty$, outage performance reaches to a floor, given by
\begin{align}
\nonumber
&P^{(i)}_{\text{out}|\frac{p}{N_{0}}\rightarrow \infty}(\gamma_{\text{th}})\\
\nonumber
&=\frac{m(m-i+1)^{1-i}}{(i-1)!(m-i)!(n-i+1)!(1-\gamma_{\text{th}}\kappa^{2}_{T})^{n-i+1}}\\
\nonumber
&\times \sum^{n-i+1}_{k=0}\binom{n-i+1}{k}(\kappa^{2}_{R}m)^{n-i-k+1}(\kappa^{2}_{T}+1)^{k}\omega^{k}\Gamma(m+k)\\
&\times (N_{0}\gamma_{\text{th}})^{n-i+1}+o\left(\left(\frac{p}{N_{0}}\right)^{-(n-i+1)}\right)\\
\nonumber
\\
\nonumber
&=\textstyle \frac{(m-i+1)^{1-i}(\kappa^{2}_{R}m)^{n+m-i+1}(N_{0}\gamma_{\text{th}})^{n-i+1}(\omega(\kappa^{2}_{T}+1))^{-m}m!}{(i-1)!(m-i)!(n-i+1)!(1-\gamma_{\text{th}}\kappa^{2}_{T})^{n-i+1}}\\
&\times \textstyle \mathcal{U}\left(m,n+m-i+2,\frac{\kappa^{2}_{R}m}{\omega(\kappa^{2}_{T}+1)}\right)+o\left(\left(\frac{p}{N_{0}}\right)^{-(n-i+1)}\right).
\label{outasympt}
\end{align}
\end{prop}

\begin{proof}
The proof is provided in Appendix \ref{appc}.
\end{proof}

\subsubsection{Imperfect CSI without hardware impairments}
The following corollary describes this simplified scenario.
\begin{cor}
Asymptotic outage floor in the case of imperfect CSI but with ideal transceiver equipment is expressed as
\begin{align}
\nonumber
&P^{(i)}_{\text{out}|\frac{p}{N_{0}}\rightarrow \infty}(\gamma_{\text{th}})=(n+m-i)!\\
&\times \frac{m(m-i+1)^{1-i}(\gamma_{\text{th}}\omega)^{n-i+1}}{(i-1)!(m-i)!(n-i+1)!}+o\left(\left(\frac{p}{N_{0}}\right)^{-(n-i+1)}\right).
\label{outasympt111}
\end{align}
\end{cor}

\begin{proof}
In the absence of hardware impairments, it holds that $P^{(i)}_{\text{out}|\frac{p}{N_{0}}\rightarrow \infty}(\gamma_{\text{th}})=\int^{\infty}_{0}F_{pr^{2}_{ii}}(p\gamma_{\text{th}}y)f_{Y_{i}}(y)dy$. Evaluating the latter integral with the aid of (\ref{Fprapprox}) yields (\ref{outasympt111}).
\end{proof}

\subsubsection{Perfect CSI with hardware impairments}
In this case, the following corollary describes the corresponding asymptotic outage performance.

\begin{cor}
Asymptotic outage performance is derived as
\begin{align}
\nonumber
&P^{(i)}_{\text{out}|\frac{p}{N_{0}}\rightarrow \infty}(\gamma_{\text{th}})=\frac{m!(m-i+1)^{1-i}}{(i-1)!(m-i)!(n-i+1)!}\\
&\times \left(\frac{\gamma_{\text{th}}\kappa^{2}_{R}m}{(1-\gamma_{\text{th}}\kappa^{2}_{T})}\right)^{n-i+1}+o\left(\left(\frac{p}{N_{0}}\right)^{-(n-i+1)}\right).
\label{outasympt1}
\end{align}
\end{cor}

\begin{proof}
Utilizing (\ref{cdf2}) and (\ref{Fprapprox}), (\ref{outasympt1}) can be readily obtained.
\end{proof}

\subsection{MMSE-SIC with Fixed Ordering}
\subsubsection{General Case}
The following proposition presents an outage floor for the general scenario of erroneous CSI under hardware impairments.
\begin{prop}
When $\frac{p}{N_{0}}\rightarrow \infty$, outage probability of the $i$th SIC stage reaches to a floor, which is given by (\ref{cdfsindrmmsesic}) and (\ref{cdfsindrmmsesicm}), when $1\leq i<m$ and $i=m$, respectively, by neglecting the $N_{0}/p$ term.
\end{prop}

The special cases of channel estimation error without hardware impairments or vice versa are obtained by setting $\kappa_{T}=\kappa_{R}=0$ or $\omega=0$, respectively.

Most importantly, the system scenario with ideal (non-impaired) hardware at the receiver provides full diversity order (i.e., $n-m+i$), regardless of the presence of imperfect CSI or the amount of hardware impairments at the transmitter. The following proposition explicitly describes this effect.
\begin{prop}
Asymptotic outage probability of the $i$th SIC stage in the presence of imperfect CSI and when hardware impairments occur only at the transmitter reads as
\begin{align}
\nonumber
&P^{(i)}_{\text{out}|\frac{p}{N_{0}}\rightarrow \infty}(\gamma_{\text{th}})=\frac{\left(\frac{N_{0}}{p\left(\kappa^{2}_{T}(\omega+1)+\omega+1\right)}\right)^{n-m+i}}{(n-m+i)!}\times\\
\nonumber
&\frac{\left(\frac{\gamma_{\text{th}}}{\gamma_{\text{th}}\left(\frac{(2\sqrt{\omega}+1)^{-1}}{\left(\kappa^{2}_{T}(\omega+1)+\omega+1\right)}-1\right)+\frac{1}{\left(\kappa^{2}_{T}(\omega+1)+\omega+1\right)}}\right)^{n}}{\left(1+\left(\frac{\gamma_{\text{th}}}{\gamma_{\text{th}}\left(\frac{(2\sqrt{\omega}+1)^{-1}}{\left(\kappa^{2}_{T}(\omega+1)+\omega+1\right)}-1\right)+\frac{1}{\left(\kappa^{2}_{T}(\omega+1)+\omega+1\right)}}\right)\right)^{m-i}}\\
&+o\left(\left(\frac{p}{N_{0}}\right)^{-(n-m+i)}\right),
\label{cdfsindrmmsesicasym}
\end{align}
and for the $m$th SIC stage as
\begin{align}
P^{(m)}_{\text{out}|\frac{p}{N_{0}}\rightarrow \infty}(\gamma_{\text{th}})=\frac{\left(\frac{\left(\frac{N_{0}\gamma_{\text{th}}}{p}\right)}{\left(1-\left(\kappa^{2}_{T}(\omega+1)+\omega\right)\gamma_{\text{th}}\right)}\right)^{n}}{n!}+o\left(\left(\frac{p}{N_{0}}\right)^{-n}\right).
\label{cdfsindrmmsesicmasym}
\end{align}
\end{prop}

\begin{proof}
The proof is provided in Appendix \ref{appf}.
\end{proof}
Notice that when $\{\omega,\kappa_{T}\}=0$, (\ref{cdfsindrmmsesicasym}) and (\ref{cdfsindrmmsesicmasym}) reflect the corresponding asymptotic outage expressions for the ideal MMSE-SIC receivers.

Collecting all the aforementioned asymptotic results, a number of conclusions can be drawn and, hence, the following remarks are outlined.

\begin{rem}
\label{rem1}
When hardware impairments and/or imperfect CSI are present, outage performance reaches to an upper bound (i.e., outage floor), regardless of the adopted equalization technique (ZF or MMSE). This is explicitly indicated in (\ref{outasympt}), (\ref{outasympt111}) and (\ref{outasympt1}) for ZF-SIC and in (\ref{cdfsindrmmsesic}) and (\ref{cdfsindrmmsesicm}) for MMSE-SIC. Therefore, there is no feasible diversity order in this case.
\end{rem}

\begin{rem}
\label{rem2}
Diversity order manifests itself, only in the case of MMSE-SIC and when there is a non-impaired receiver, regardless of the presence of hardware impairments at the transmitter and/or imperfect CSI at the receiver. This is indicated in (\ref{cdfsindrmmsesicasym}) and (\ref{cdfsindrmmsesicmasym}), where both expressions tend to zero as $p/N_{0}\rightarrow \infty$ (by noticing the existence of the $N_{0}/p$ term within these expressions). Particularly, the diversity order in this case is $n-i+1$ with respect to the $i$th decoding layer or $n-m+i$ with respect to the $i$th SIC stage.
\end{rem}

It can be easily seen that the latter remark indicates no difference in the diversity order of the considered MMSE-SIC and the classical MMSE-SIC of an ideal communication setup (see, e.g., \cite{ref16}). Apparently, performance difference between these two scenarios appears to the underlying coding (array) gains. Observe that ZF-SIC does not achieve diversity order, even when hardware impairments occur only at the transmitter. This effect occurs due to the fact that ZF, in principle, operates by fully eliminating interference but enhancing the noise at the same time. When noise power is proportional to the transmission power, then it unavoidably reflects to the aforementioned outage floor. Such observations could be quite useful for system designers of various MIMO practical applications. As an indicative example, it is preferable to enable higher quality hardware gear for the antennas of the receiver rather than the transmitter. When such a condition occurs, the performance difference of MMSE-SIC over ZF-SIC is emphatically increased for larger SINDR regions. Yet, in order to achieve this performance gain, the variances of channel estimation error and hardware impairments at the transceiver are required, i.e., see the linear filter in (\ref{filterg}).

\section{Error Propagation Effect}
One of the most important degradation factors of SIC-based reception is the well-known error propagation effect. To date, it has been studied mainly numerically (e.g., see \cite{ref13} and references therein) and semi-analytically \cite{ref7} in terms of integral or bound expressions. The limited scenario of $m=2$ was analytically studied in \cite{ref22}, but the derived expressions were in terms of infinite series representations. In this section, error propagation is analyzed with regards to the average symbol error probability (ASEP). A formula including numerical verifications is presented for the general case, while closed-form expressions are obtained for some special cases of interest.

ASEP of the $i$th decoding layer, namely ASEP$_{i}$, explicitly reads as
\begin{align}
\nonumber
\text{ASEP}_{i}&\triangleq \text{Pr}\left[\epsilon_{i}|\epsilon_{m}\right]\text{Pr}\left[\epsilon_{m}\right]\\
\nonumber
&\ \ \ \ +\text{Pr}\left[\epsilon_{i}|\epsilon_{m-1}\cap\epsilon^{c}_{m}\right]\text{Pr}\left[\epsilon_{m-1}\cap \epsilon^{c}_{m}\right]+\cdots\\
\nonumber
&\ \ \ \ +\text{Pr}\left[\epsilon_{i}|\epsilon_{i+1}\cap\left(\bigcap^{m}_{l=i+2}\epsilon^{c}_{l}\right)\right]\\
\nonumber
&\ \ \ \ \times \text{Pr}\left[\epsilon_{i+1}\cap\left(\bigcap^{m}_{l=i+2}\epsilon^{c}_{l}\right)\right]\\
\nonumber
&\ \ \ \ +\text{Pr}\left[\epsilon_{i}|\bigcap^{m}_{l=i+1}\epsilon^{c}_{l}\right]\text{Pr}\left[\bigcap^{m}_{l=i+1}\epsilon^{c}_{l}\right]\\
&=\left(1-\frac{1}{\mathcal{M}}\right)\sum^{m}_{t=i}\text{Pr}\left[\epsilon_{t}|\bigcap^{m}_{l=t+1}\epsilon^{c}_{l}\right]\text{Pr}\left[\bigcap^{m}_{l=t+1}\epsilon^{c}_{l}\right],
\label{asepi}
\end{align}
where $\epsilon_{i}$ denotes an error event at the $i$th decoding layer, $\epsilon^{c}_{i}$ is the complement of $\epsilon_{i}$, while $\mathcal{M}$ represents the number of modulation states. Also, the second equality of (\ref{asepi}) arises by assuming that an earlier error (with probability one) results in a uniform distribution over the constellation for a subsequent symbol decision (equal-power constellation).

Hence, ASEP describing the overall behavior of the system, namely $\overline{\text{ASEP}}$, is given by
\begin{align}
\nonumber
\overline{\text{ASEP}}&=\frac{1}{m}\sum_{i}\text{ASEP}_{i}\\
&=\frac{\left(1-\frac{1}{\mathcal{M}}\right)}{m}\sum^{m}_{t=1}t\overline{P}_{s_{t}}\prod^{m}_{l=t+1}\left(1-\overline{P}_{s_{l}}\right),
\label{asep}
\end{align}
where $\overline{P}_{s_{i}}\triangleq \text{Pr}\left[\epsilon_{i}|\bigcap^{m}_{l=i+1}(1-\epsilon_{l})\right]$ is the conditional ASEP at the $i$th decoding layer given that there are no errors in prior layers.

Thereby, finding $\overline{P}_{s_{i}}$ represents a key issue to prescribe the total ASEP. It holds that \cite{ref15}
\begin{align}
\overline{P}_{s_{i}}\triangleq \frac{\mathcal{A}\sqrt{\mathcal{B}}}{2\sqrt{\pi}}\int^{\mathcal{Z}}_{0}\frac{\exp(-\mathcal{B} x)}{\sqrt{x}}P^{(i)}_{\text{out}}(x)dx,
\label{asepdef}
\end{align}
where $\mathcal{Z}=1/\kappa^{2}_{T}$ for ZF-SIC, while $\mathcal{Z}=1/(\kappa^{2}_{R}(\omega+1)+\omega)$ for MMSE-SIC. Note that $\mathcal{Z}\rightarrow +\infty$, in ideal conditions of both schemes. Also, $\mathcal{A}$ and $\mathcal{B}$ are specific constants that define the modulation type \cite{ref1555}.

Unfortunately, there is no straightforward closed-form solution for $\overline{P}_{s_{i}}$ for the general case of ZF-SIC and MMSE-SIC, which is based on (\ref{outclosed}), (\ref{cdfsindrmmsesic}) and (\ref{cdfsindrmmsesicm}), to our knowledge. Thus, $\overline{P}_{s_{i}}$ and $\overline{\text{ASEP}}$ can be resolved only via numerical methods. Still, the involvement of a single numerical integration is much more efficient than classical simulation methods (e.g., Monte-Carlo). In the following, some certain scenarios of special interest admit a closed formulation of $\overline{P}_{s_{i}}$, which in turn provide a corresponding solution to $\overline{\text{ASEP}}$.

\subsection{Ordered ZF-SIC}

\begin{prop}
The closed-form expression for $\overline{P}_{s_{i}}$ in the presence of channel estimation errors, an impaired receiver and an ideal transmitter is derived as
\begin{align}
\nonumber
&\overline{P}_{s_{i}}\approx \frac{\mathcal{A}}{2}\Bigg[1-\sqrt{\frac{\mathcal{B}}{\pi}}\Psi_{i}\sum^{\mu}_{v=0}\binom{\mu}{v}(v+m-1)!\\
\nonumber
&\times \frac{p^{v}(p\kappa^{2}_{R}m+N_{0})^{\mu-v}\Gamma(\mu+\frac{1}{2})}{\Gamma(m)\omega^{\mu+\frac{1}{2}-v}(m+l-i+1)^{\mu+\frac{1}{2}}}\\
&\times \mathcal{U}\left(\textstyle \mu+\frac{1}{2},\mu+\frac{3}{2}-v-m,\frac{(p\kappa^{2}_{R}m+N_{0})}{p\omega}+\frac{\mathcal{B}}{\omega(m+l-i+1)}\right)\Bigg].
\label{aseppp}
\end{align}
\end{prop}

\begin{proof}
Plugging (\ref{outclosed}) in (\ref{asepdef}), setting $\kappa_{T}=0$, while utilizing \cite[Eq. (2.3.6.9)]{ref11}, gives (\ref{aseppp}).
\end{proof}

Notice that although (\ref{aseppp}) is involved with a special function (i.e., Tricomi confluent hypergeometric function), it is in a form of finite sum series, whereas is included as standard built-in function in several popular mathematical software packages. Hence, this expression can be easily and efficiently calculated.\footnote{The asymptotic ASEP expressions could be easily extracted, by following the same methodology as in the previous section. Yet, they have omitted herein since they present very similar insights as the previously derived asymptotic outage probabilities.}

\subsection{MMSE-SIC with Fixed Ordering}

\begin{prop}
$\overline{P}_{s_{i}}$, for the $i$th SIC stage ($1\leq i< m$), in the presence of perfect CSI, a non-impaired transmitter, and an impaired receiver is expressed as
\begin{align}
\nonumber
&\overline{P}_{s_{i}}=\frac{\mathcal{A}}{2}\vast\{1-\sqrt{\frac{\mathcal{B}}{\pi}}\Bigg[\sum^{n}_{k_{1}=1}\frac{\Gamma(k_{1}-\frac{1}{2})\left(\kappa^{2}_{R}m+\frac{N_{0}}{p}\right)^{k_{1}-1}}{(k_{1}-1)!\left(\kappa^{2}_{R}m+\frac{N_{0}}{p}+\mathcal{B}\right)^{k_{1}-\frac{1}{2}}}\\
\nonumber
&-\sum^{n}_{k_{2}=n-m+i+1}\:\:\sum^{m-i}_{j=n-k_{2}+1}\binom{m-i}{j}\left(\kappa^{2}_{R}m+\frac{N_{0}}{p}\right)^{k_{2}-1}\\
&\times \frac{\Gamma\left(\scriptstyle k_{2}+j-\frac{1}{2}\right)}{(k_{2}-1)!}\mathcal{U}\left(\scriptstyle k_{2}+j-\frac{1}{2},k_{2}+j+i-m+\frac{1}{2},\mathcal{B}+\kappa^{2}_{R}m+\frac{N_{0}}{p}\right)\Bigg]\vast\}.
\label{asepmmsei}
\end{align}
\end{prop}

\begin{proof}
By invoking (\ref{cdfsindrmmsesic}) in (\ref{asepdef}), setting $\{\kappa_{T},\omega\}=0$, while utilizing \cite[Eq. (2.3.6.9)]{ref11}, (\ref{asepmmsei}) is obtained.
\end{proof}

For $i=m$, in the last SIC stage, the expression of (\ref{cdfsindrmmsesicm}) does not admit a closed formulation of ASEP. However, it can be numerically calculated quite easily by using (\ref{cdfsindrmmsesicm}) in (\ref{asepdef}) over the valid integration range $\{0,\frac{1}{(\kappa^{2}_{T}(\omega+1)+\omega)}\}$.

\section{Numerical Results}
\begin{figure}[!t]
\centering
\includegraphics[keepaspectratio,width=\columnwidth]{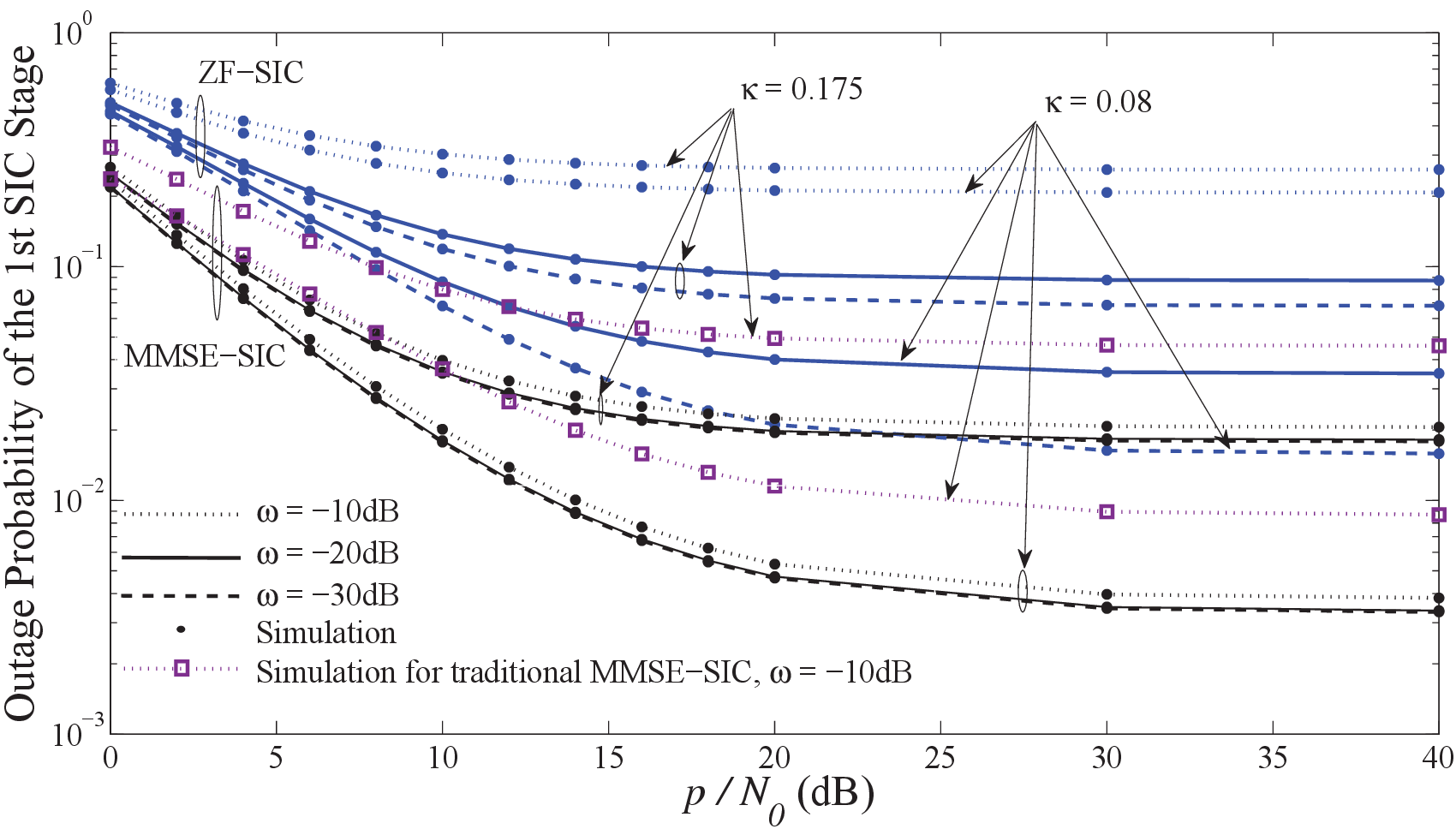}
\caption{Outage performance of the 1st SIC stage (i.e., the $m$th decoding layer) of the ordered ZF-SIC and unordered (fixed) MMSE-SIC vs. various average input $p/N_{0}$ values, where $\left\{n,m\right\}=4$ and $\gamma_{\text{th}}=0$dB.}
\label{fig1}
\end{figure}
\begin{figure}[!t]
\centering
\includegraphics[keepaspectratio,width=\columnwidth]{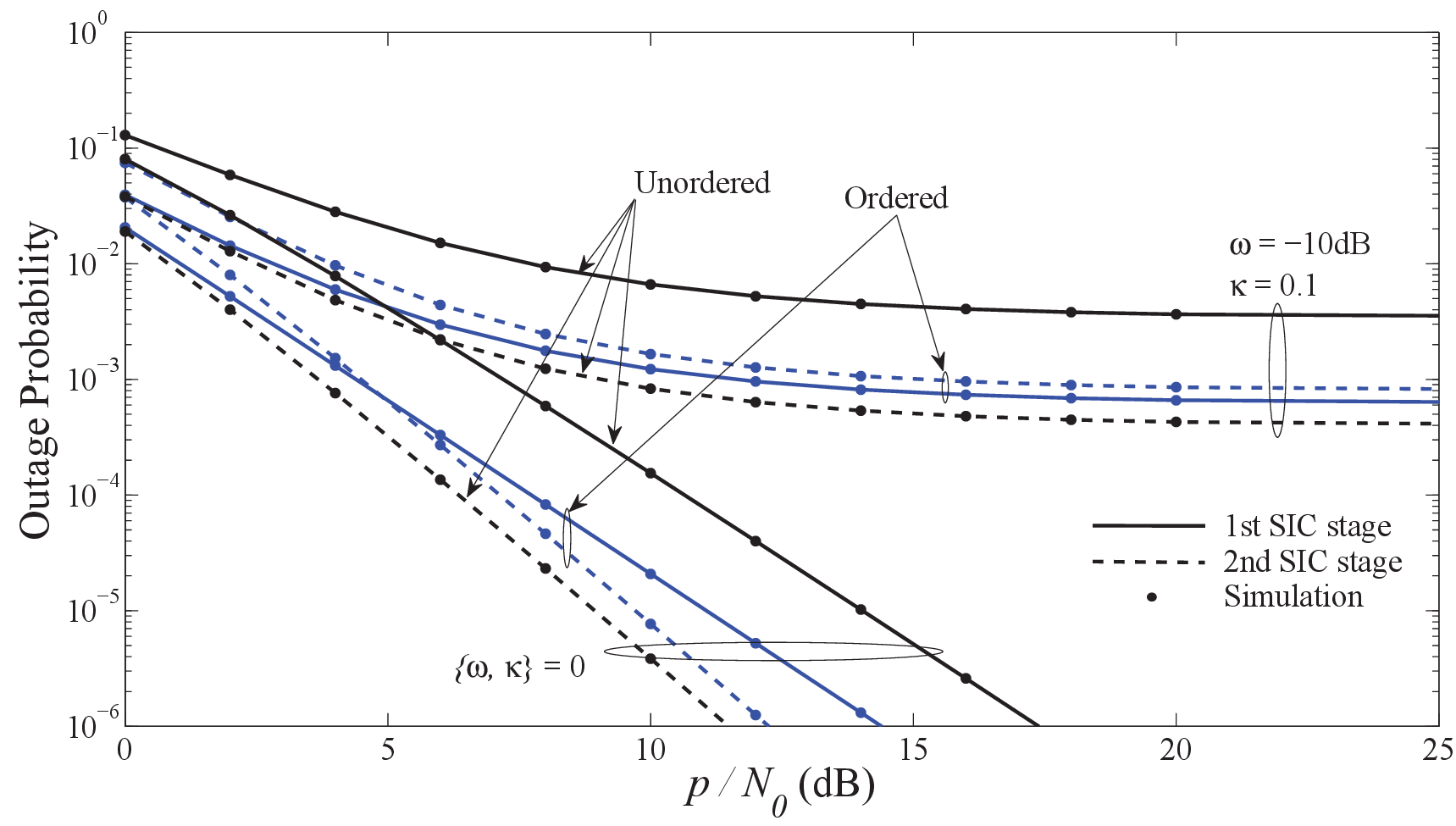}
\caption{Outage performance of each SIC stage for the ordered and unordered ZF-SIC vs. various average input $p/N_{0}$ values, where $n=4$, $m=2$ and $\gamma_{\text{th}}=0$dB.}
\label{fig2}
\end{figure}

In this section, analytical results are presented and cross-compared with Monte-Carlo simulations. There is a good match between all the analytical and the respective simulation results and, hence, the accuracy of the proposed approach is verified. Note that in Figs. \ref{fig1} and \ref{fig2}, for ease of tractability and without loss of generality, we assume symmetric levels of impairments at the transceiver, i.e., an equal hardware quality at the transmitter and receiver. To this end, let $\kappa_{T}=\kappa_{R}\triangleq \kappa$.

In Fig. \ref{fig1}, the outage performance for the 1st stage of the ordered ZF- and unordered MMSE-SIC is presented for various system settings/conditions. There is an emphatic performance difference between the two schemes in all the considered cases, despite the fact that no optimal ordering is used in MMSE-SIC. This observation verifies the superiority of MMSE against ZF detectors in non-ideal communication setups. In addition, it is obvious that CSI imperfection impacts the performance of ZF-SIC in greater scale than hardware impairments. When this imperfection is more relaxed, the performance gap between the two extreme hardware impairment scenarios starts to grow. This occurs because ZF, fundamentally, relies on channel estimation accuracy to achieve performance gains, counteracting the unavoidable noise enhancement. Thereby, CSI imperfection dramatically affects its performance in comparison to the (noise-oriented) hardware imperfection. Interestingly, this does not comply with MMSE-SIC, whereas quite the opposite condition holds. This is consistent with Remark \ref{rem2}. Also, the traditional MMSE-SIC scheme (taking into consideration only the channel gains and $N_{0}$) is included for performance comparison reasons. The performance gain of the presented MMSE-SIC over its traditional counterpart is straightforward. 

Figure \ref{fig2} depicts the ordered and unordered outage performance of ZF-SIC in ideal and non-ideal communication setups. Obviously, diversity order is manifested only in the former case, while an outage floor is presented in the latter case. This is consistent with Remark \ref{rem1}. It is also noteworthy that the diversity order remains unaffected from the ordering strategy, in accordance to \cite{ref31}. Moreover, the superiority of the ordered 1st SIC stage against the corresponding unordered stage can be clearly seen. This is the price of performing optimal detection ordering. Furthermore, an important observation from the non-ideal scenario is the fact that the 2nd stage has worse performance as compared to the 1st stage of the ordered ZF-SIC in the entire SNR region. This should not be confusing since the 2nd stage of the ordered SIC has always the worst SNR, whereas this is not the case for the unordered SIC (on average). It seems that less interference (at the 2nd stage) is not enough to counteract the presence of channel imperfection severity and impaired hardware and, hence, to outperform 1st stage. This is in contrast to the traditional (ideal) SIC receivers, where the 1st SIC stage influences more drastically the overall system performance, representing a lower outage performance bound \cite{ref22,ref32}.
\begin{figure}[!t]
\centering
\includegraphics[keepaspectratio,width=\columnwidth]{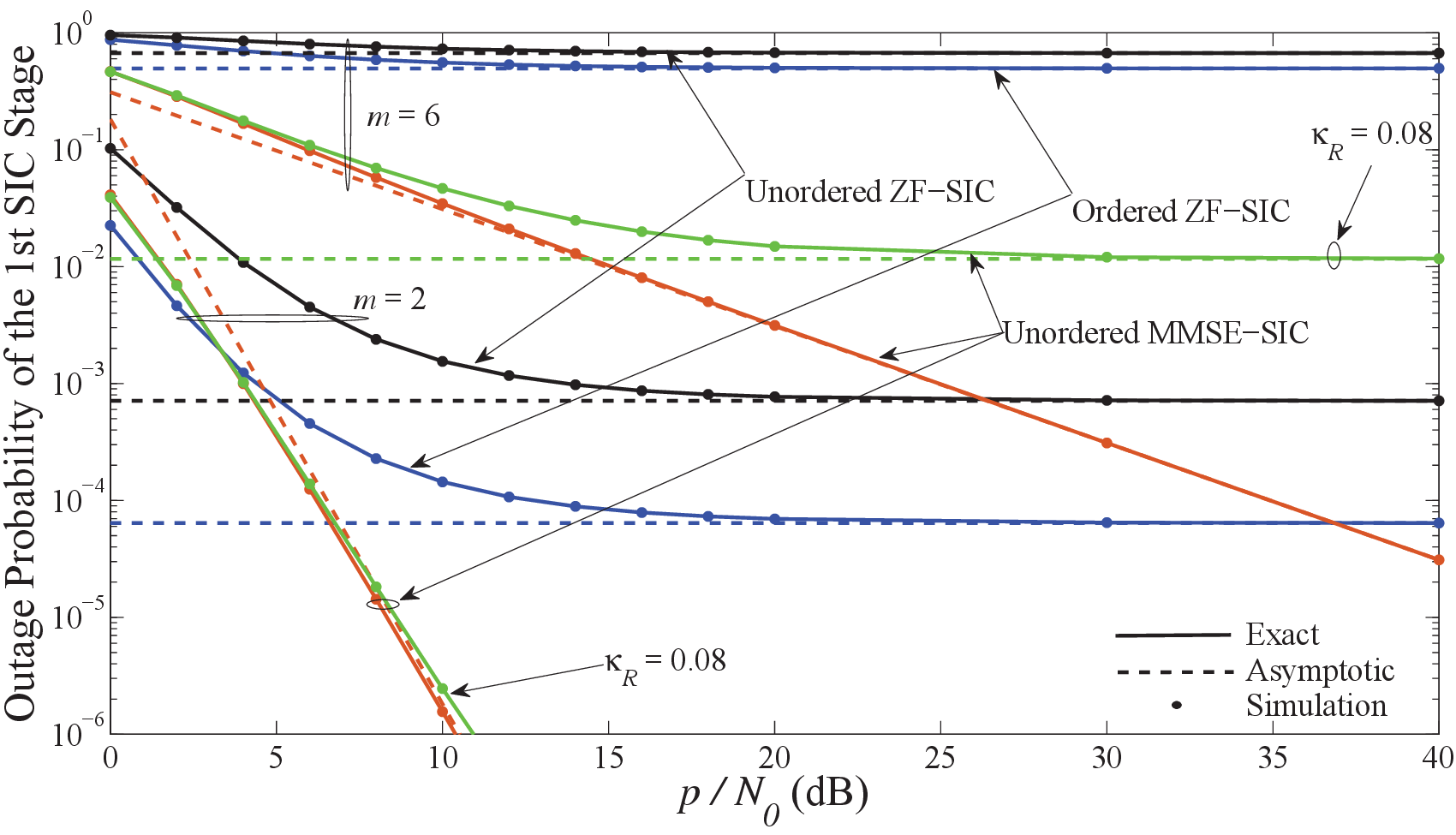}
\caption{Outage performance of the 1st SIC stage of the ordered ZF-SIC, unordered ZF-SIC and unordered (fixed) MMSE-SIC vs. various average input $p/N_{0}$ values, where $n=6$, $\gamma_{\text{th}}=3$dB, $\omega=-10$dB, $\kappa_{T}=0.08$, and $\kappa_{R}=0$ (unless stated otherwise).}
\label{fig3}
\end{figure}
Figure \ref{fig3} highlights the important outcome of Remark \ref{rem2} in non-ideal communication systems. Specifically, it can be seen that when hardware impairments occur only at the transmitter side, MMSE-SIC maintains its diversity order, while ZF-SIC introduces an outage floor, confirming the previous analysis. Also, in dense multi-stream transmissions (i.e., when $m=6$), outage performance of ZF-SIC is rather inefficient in comparison to MMSE-SIC.

ASEP of the 1st MMSE-SIC stage is presented in Fig. \ref{fig4} for various settings, using (\ref{asepdef}). Again, it is verified that providing a higher-cost/higher-quality hardware gear at the receiver side is a much more fruitful option. Finally, Fig. \ref{fig5} presents the overall ASEP using (\ref{asep}), for the two considered SIC schemes. All the results for the ZF-SIC are obtained using (\ref{aseppp}). In addition, the corresponding results of MMSE-SIC for the scenarios with imperfect and perfect CSI are obtained via numerical integration (as in Fig. \ref{fig4}) and using (\ref{asepmmsei}), respectively.
\begin{figure}[!t]
\centering
\includegraphics[keepaspectratio,width=\columnwidth]{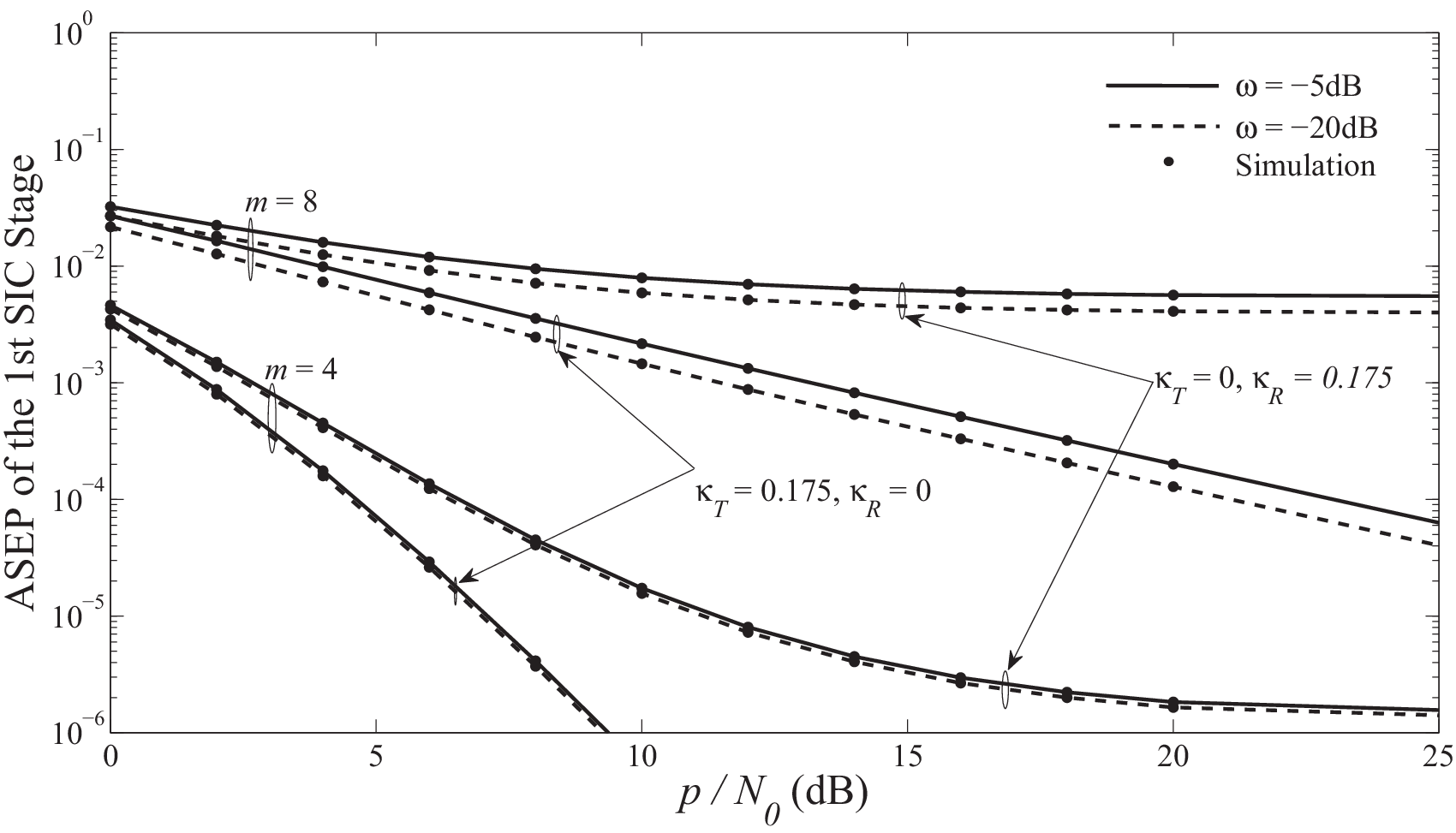}
\caption{ASEP of the 1st SIC stage for MMSE-SIC with fixed ordering under a BPSK modulation scheme vs. various average input $p/N_{0}$ values, where $n=8$.}
\label{fig4}
\end{figure}
\begin{figure}[!t]
\centering
\includegraphics[keepaspectratio,width=\columnwidth]{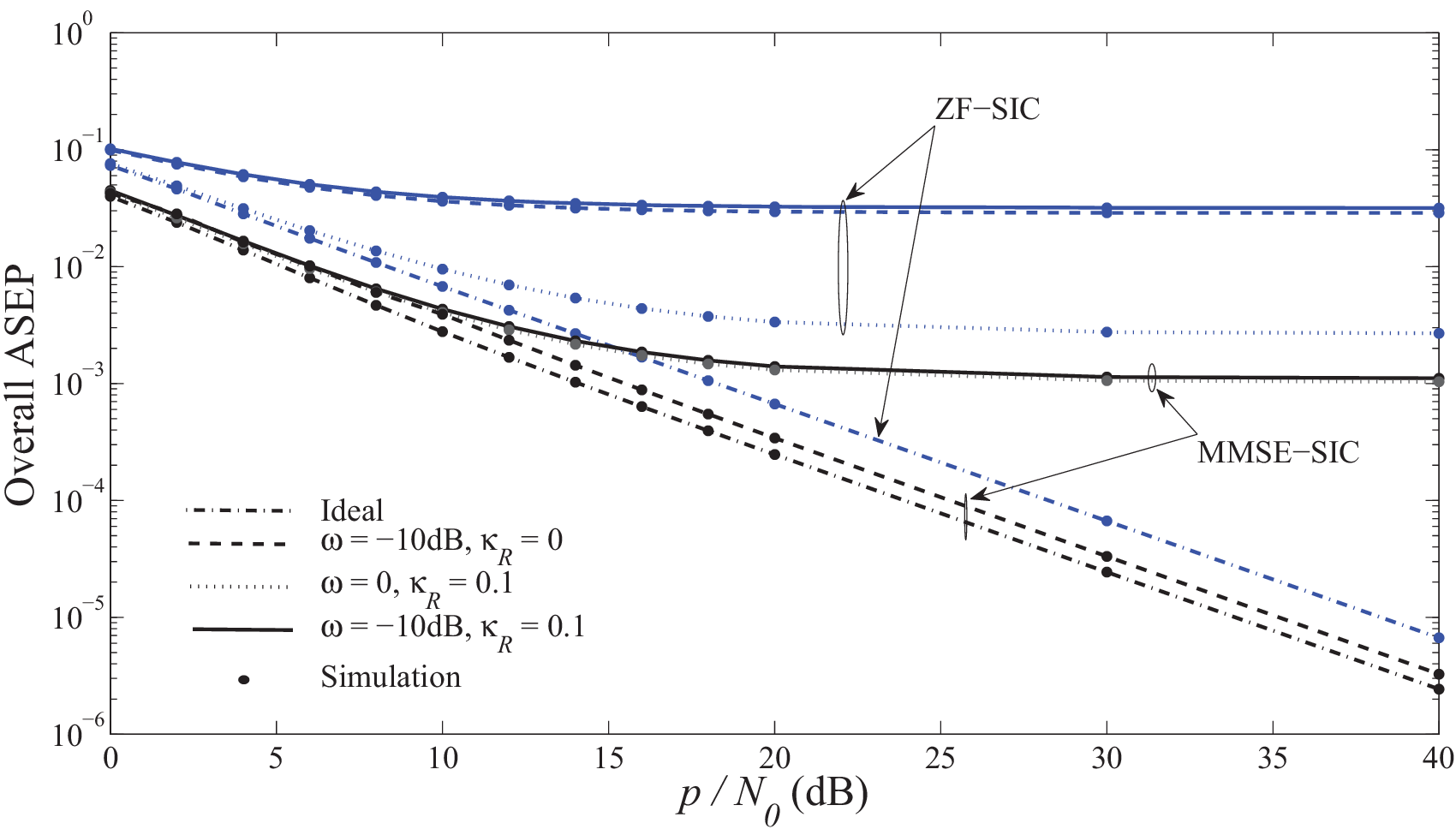}
\caption{Total ASEP of the ordered ZF-SIC and unordered (fixed) MMSE-SIC under a BPSK modulation scheme vs. various average input $p/N_{0}$ values, where $\{n,m\}=4$ and $\kappa_{T}=0$.}
\label{fig5}
\end{figure}
Considering all the above, both the outage and error rate numerical results confirm the theoretical framework, while the following important outcomes are summarized: a) In the case of ZF-SIC, hardware impairments at the transmitter are as crucial (proportionally) as the impairments at the receiver; b) in ZF-SIC schemes, CSI imperfection influences more the performance than hardware impairments; c) MMSE-SIC appropriately counterbalance the impact of CSI imperfection and the amount of impaired hardware; d) when $\kappa_{R}=0$, MMSE-SIC maintains diversity order and, thus, there is an emphatic performance gain over ZF-SIC, especially in medium-to-high SNR regions.

\section{Conclusions}
Successive decoding of multiple individual streams was thoroughly investigated under practical communication scenarios. Particularly, ZF-SIC detection/decoding with symbol ordering and MMSE-SIC with fixed ordering were studied for hardware-impaired transceivers and when CSI is imperfectly provided at the receiver side. The analysis included i.i.d. Rayleigh multipath fading channels. New analytical and quite simple (in terms of computational complexity) expressions regarding the outage probability for each SIC stage were obtained. In addition, a general formula indicating the error rate performance with regards to the error propagation effect is provided. Moreover, it was indicated that MMSE-SIC outperforms ZF-SIC in non-ideal communication systems in spite of utilizing no optimal ordering. In addition, an unavoidable performance floor is introduced in the general scenario for both schemes, while diversity order is maintained in MMSE-SIC only when an ideal hardware equipment is enabled at the receiver.

\appendix
\subsection{Derivation of (\ref{filterg}), (\ref{sindrmmse}) and (\ref{sindrmmsesicm})}
\label{app0}
\numberwithin{equation}{subsection}
\setcounter{equation}{0}
From (\ref{mse}), we have that
\begin{align}
\nonumber
&\text{MSE}^{(j)}\\
\nonumber
&=(\mathbf{g}^{(j)})^{\mathcal{H}}\mathcal{\mathbf{C}}\mathbf{g}^{(j)}-p(\mathbf{g}^{(j)})^{\mathcal{H}}\mathbf{h}_{j}-p\mathbf{h}_{j}^{\mathcal{H}}\mathbf{g}^{(j)}+p\\
\nonumber
&=\left((\mathbf{g}^{(j)})^{\mathcal{H}}-p\mathbf{h}_{j}^{\mathcal{H}}\mathcal{\mathbf{C}}^{-1}\right)\mathcal{\mathbf{C}}\left((\mathbf{g}^{(j)})^{\mathcal{H}}-p\mathbf{h}_{j}^{\mathcal{H}}\mathcal{\mathbf{C}}^{-1}\right)^{\mathcal{H}}\\
&\ \ +p-p^{2}\mathbf{h}_{j}^{\mathcal{H}}\mathcal{\mathbf{C}}^{-1}\mathbf{h}_{j}^{\mathcal{H}},
\label{aaa}
\end{align}
where $\mathcal{\mathbf{C}}\triangleq \mathbb{E}[\mathbf{y}\mathbf{y}^{\mathcal{H}}]=p\mathbf{H}\mathbf{H}^{\mathcal{H}}+\mathbb{E}[\mathbf{w'}\mathbf{w'}^{\mathcal{H}}]$. Since only the first term of (\ref{aaa}) depends on $\mathbf{g}^{(j)}$, the optimal solution that minimizes MSE is $\mathbf{g}^{(j)}=p\mathcal{\mathbf{C}}^{-1}\mathbf{h}_{j}$ and, hence, (\ref{filterg}) is obtained.

At the receiver, $(\mathbf{g}^{(j)})^{\mathcal{H}}\mathbf{y}$ is performed, yielding
\begin{align*}
z_{j}=(\mathbf{g}^{(j)})^{\mathcal{H}}\mathbf{y}=\beta_{j}s_{j}+\eta_{j},
\end{align*}
where $\beta_{j}\triangleq (\mathbf{g}^{(j)})^{\mathcal{H}}\mathbf{h}_{j}$ and $\eta_{j}\triangleq \sum_{l\neq j}(\mathbf{g}^{(j)})^{\mathcal{H}}\mathbf{h}_{l}s_{l}+(\mathbf{g}^{(j)})^{\mathcal{H}}\mathbf{w}'$. Then, the variance of $\eta_{j}$ is computed as
\begin{align}
\nonumber
&\mathbb{E}[\eta_{j}\eta^{\mathcal{H}}_{j}]=p(2\sqrt{\omega}+1)(\mathbf{g}^{(j)})^{\mathcal{H}}\mathbf{K}_{j}\mathbf{K}_{j}^{\mathcal{H}}\mathbf{g}^{(j)}+(\mathbf{g}^{(j)})^{\mathcal{H}}\mathbf{H}\mathbf{H}^{\mathcal{H}}\\
\nonumber
&\ \  \times(p\kappa^{2}_{T}(\omega+1)+p\omega)\mathbf{g}^{(j)}+(\mathbf{g}^{(j)})^{\mathcal{H}}(p\kappa^{2}_{R}m+N_{0})\mathbf{g}^{(j)}\\
\nonumber
&=p(2\sqrt{\omega}+1)(\mathbf{g}^{(j)})^{\mathcal{H}}\mathbf{K}_{j}\mathbf{K}_{j}^{\mathcal{H}}\mathbf{g}^{(j)}+(\mathbf{g}^{(j)})^{\mathcal{H}}\mathbf{H}\mathbf{H}^{\mathcal{H}}\\
\nonumber
&\ \  \times(p\kappa^{2}_{T}(\omega+1)+p\omega)\mathbf{g}^{(j)}+(\mathbf{g}^{(j)})^{\mathcal{H}}(p\kappa^{2}_{R}m+N_{0})\mathbf{g}^{(j)}\\
\nonumber
&\  \ +\frac{p}{2\sqrt{\omega}+1}((\mathbf{g}^{(j)})^{\mathcal{H}}\mathbf{h}_{j})^{2}-\frac{p}{2\sqrt{\omega}+1}((\mathbf{g}^{(j)})^{\mathcal{H}}\mathbf{h}_{j})^{2}\\
\nonumber
&= p(\mathbf{g}^{(j)})^{\mathcal{H}}\bigg(\mathbf{H}\mathbf{H}^{\mathcal{H}}\left(\kappa^{2}_{T}(\omega+1)+\omega+1\right)\\
\nonumber
&\  \ +\left(\kappa^{2}_{R}m+\frac{N_{0}}{p}\right)\mathbf{I}_{n}\bigg)\mathbf{g}^{(j)}-\frac{p}{2\sqrt{\omega}+1}\beta^{2}_{j}\\
&=p\left(\beta_{j}-\frac{1}{2\sqrt{\omega}+1}\beta^{2}_{j}\right),
\label{variancemmse}
\end{align}
where $\mathbf{K}_{j}\triangleq [\mathbf{h}_{1} \cdots \mathbf{h}_{j-1}\:\: \mathbf{h}_{j+1}\cdots \mathbf{h}_{m}]$. Thus, the SINDR of the $j$th stream is given by
\begin{align*}
\text{SINDR}^{(j)}=\frac{p\beta^{2}_{j}}{\mathbb{E}[\eta_{j}\eta^{\mathcal{H}}_{j}]}=\frac{\beta_{j}}{1-\frac{\beta_{j}}{2\sqrt{\omega}+1}},\ \ 0<\beta_{j}<1,
\end{align*}
and thus we arrive at (\ref{sindrmmse}).

At the MMSE-SIC receiver, the corresponding SINDR expression is presented in (\ref{sindrmmsesic}) for the $i$th SIC stage, when $i<m$ (i.e., except the final SIC stage). At the last SIC stage, when $i=m$, there is no residual interference caused by other streams. Following the same methodology as in (\ref{variancemmse}), we have that
\begin{align}
\nonumber
&\mathbb{E}[\eta_{m}\eta^{\mathcal{H}}_{m}]=p(\mathbf{g}^{(m)})^{\mathcal{H}}\\
&\times \bigg(\mathbf{h}\mathbf{h}^{\mathcal{H}}\left(\kappa^{2}_{T}(\omega+1)+\omega\right)+\left(\kappa^{2}_{R}m+\frac{N_{0}}{p}\right)\mathbf{I}_{n}\bigg)\mathbf{g}^{(m)},
\label{variancemmse1}
\end{align}
while the corresponding SINDR stems in (\ref{sindrmmsesicm}).

\subsection{Derivation of (\ref{outclosed})}
\label{appa}
\numberwithin{equation}{subsection}
\setcounter{equation}{0}
Based on (\ref{pdfrii}) and utilizing \cite[Eq. (3.351.2)]{ref1}, we have that
\begin{align}
\nonumber
\text{Pr}\left[p r^{2}_{ii}\geq z\right]&=\int^{\infty}_{z}f_{pr^{2}_{ii}}(x)dx\\
&=\Psi_{i} z^{\mu} \exp\left(-\frac{(m+l-i+1)z}{p}\right).
\label{outappen}
\end{align}
Thus, based on (\ref{cdf1}), the unconditional CDF of SINDR for the $i$th decoding layer is expressed as
\begin{align}
\nonumber
&F_{\text{SINDR}_{i}}(\gamma_{\text{th}})\approx\\
&\int^{\infty}_{0}F_{p r^{2}_{ii}|Y_{i}}\left(\textstyle \frac{\left(p\left(\kappa^{2}_{T}+1\right)y+p\kappa^{2}_{R}m+N_{0}\right)\gamma_{\text{th}}}{\left(1-\kappa^{2}_{T}\gamma_{\text{th}}\right)}\bigg|y\right) f_{Y_{i}}(y)dy,
\label{outappen1}
\end{align}
where $F_{X|Y}(.)$ denotes the conditional CDF of $X$ given $Y$. Then, plugging (\ref{cdf1}) and (\ref{pdferror}) into (\ref{outappen1}), using the binomial expansion \cite[Eq. (1.111)]{ref1} and the integral identity \cite[Eq. (3.351.3)]{ref1}, (\ref{outclosed}) is easily derived after some straightforward manipulations.

\subsection{Derivation of (\ref{outclosed1})}
\label{appb}
Referring back to (\ref{refer}) and neglecting the $\mathbf{\Delta H}$ term, (\ref{sindr}) becomes
\begin{align}
\text{SNDR}_{i}=\frac{p r^{2}_{ii}}{pr^{2}_{ii}\kappa^{2}_{T}+p\kappa^{2}_{R}m+N_{0}},
\label{sindr1}
\end{align}
where SNDR stands for the signal-to-noise-plus-distortion ratio. Hence, following similar lines of reasoning as for the derivation of (\ref{cdf1}), the (unconditional) CDF of SNDR yields as
\begin{align}
F_{\text{SNDR}_{i}}(\gamma_{\text{th}})=1-\text{Pr}\left[p r^{2}_{ii}\geq \frac{\left(p\kappa^{2}_{R}m+N_{0}\right)\gamma_{\text{th}}}{\left(1-\kappa^{2}_{T}\gamma_{\text{th}}\right)}\right].
\label{cdf2}
\end{align}
Using (\ref{outappen}) in (\ref{cdf2}), we have that
\begin{align}
\nonumber
F_{\text{SNDR}_{i}}(\gamma_{\text{th}})&=1-\Psi_{i} \left(\frac{\left(p\kappa^{2}_{R}m+N_{0}\right)\gamma_{\text{th}}}{\left(1-\kappa^{2}_{T}\gamma_{\text{th}}\right)}\right)^{\mu}\\
&\times \exp\left(-\frac{(m+l-i+1)\left(p\kappa^{2}_{R}m+N_{0}\right)\gamma_{\text{th}}}{p\left(1-\kappa^{2}_{T}\gamma_{\text{th}}\right)}\right).
\label{cdf3}
\end{align}
Finally, recognizing that $P^{(i)}_{\text{out}}(\gamma_{\text{th}})\triangleq F_{\text{SNDR}_{i}}(\gamma_{\text{th}})$, the proof is completed.

\subsection{Derivation of (\ref{cdfsindrmmsesic}) and (\ref{cdfsindrmmsesicm})}
\label{appd}
Recall that $P^{(i)}_{\text{out}}(\gamma_{\text{th}})\triangleq F_{\text{SINDR}_{i}}(\gamma_{\text{th}})$. Also, observe from (\ref{sindrmmsesic}) that for the $i$th SIC stage ($1\leq i<m$), we have that
\begin{align}
\nonumber
&\text{Pr}[\text{SINDR}_{i}\leq \gamma_{\text{th}}]=\\
&\text{Pr}\Bigg[\Phi_{i} \leq \frac{\gamma_{\text{th}}}{\gamma_{\text{th}}\left(\frac{(2\sqrt{\omega}+1)^{-1}}{(\kappa^{2}_{T}(\omega+1)+\omega+1}-1\right)+\frac{1}{(\kappa^{2}_{T}(\omega+1)+\omega+1}}\Bigg],
\label{phi}
\end{align}
where $\Phi_{i} \triangleq \mathbf{h}_{i}^{\mathcal{H}}(\hat{\mathbf{K}_{i}}\hat{\mathbf{K}_{i}}^{\mathcal{H}}+\frac{(\kappa^{2}_{R}m+N_{0}/p)}{(\kappa^{2}_{T}(\omega+1)+\omega+1)}\mathbf{I}_{n})^{-1}\mathbf{h}_{i}$. In order for (\ref{phi}) to be a valid CDF, $\gamma_{\text{th}}<1/(\kappa^{2}_{T}(\omega+1)+\omega)$ is required. Otherwise, $\text{Pr}[\text{SINDR}_{i}\leq \gamma_{\text{th}}]=1$.

Fortunately, based on the pioneer work in \cite[Eq. (11)]{ref25}, and some further elaborations on this result (e.g., see \cite[Eq. (6)]{ref26} and \cite[Eq. (61)]{ref27}), CDF of $\Phi_{i}$ yields as
\begin{align}
\nonumber
F_{\Phi_{i}}(x)&=1-\exp\left(-\frac{(\kappa^{2}_{R}m+N_{0}/p)x}{(\kappa^{2}_{T}(\omega+1)+\omega+1)}\right)\\
&\times \sum^{n}_{k=1}\frac{A_{k}(x)\left(\frac{(\kappa^{2}_{R}m+N_{0}/p)x}{(\kappa^{2}_{T}(\omega+1)+\omega+1)}\right)^{k-1}}{(k-1)!},
\label{cdfmmsesic1}
\end{align}
where
\begin{align*}
A_{k}(x)=\left\{
\begin{array}{c l}
    1, & n\geq m+k-i,\\
    & \\
    \frac{1+\sum^{n-k}_{j=1}\binom{m-i}{j}x^{j}}{\left(1+x\right)^{m}}, & n < m+k-i.
\end{array}\right.
\end{align*}
Nonetheless, (\ref{cdfmmsesic1}) is quite cumbersome and it is not amenable for further analysis. Due to this, we slightly modify it in order to derive a more convenient formation. Noticing that $n\geq m$, using (\ref{phi}), and the fact that
\begin{align*}
\frac{1+\sum^{n-k}_{j=1}\binom{m-i}{j}x^{j}}{\left(1+x\right)^{m-i}}=1-\frac{\sum^{m-i}_{j=n-k+1}\binom{m-i}{j}x^{j}}{\left(1+x\right)^{m-i}},
\end{align*}
we arrive at (\ref{cdfsindrmmsesic}), after some simple manipulations.

At the last SIC stage ($i=m$), based on (\ref{sindrmmsesicm}), it holds that
\begin{align}
\text{SINDR}_{m}\overset{\text{d}}=\frac{\left(\frac{1}{(\kappa^{2}_{R}m+N_{0}/p)}\mathcal{Y}\right)}{\left(\frac{(\kappa^{2}_{T}(\omega+1))+\omega}{(\kappa^{2}_{R}m+N_{0}/p)}\mathcal{Y}+1\right)},
\end{align}
where $\mathcal{Y}\triangleq \sum^{n}_{l=1}\left|h_{l}\right|^{2}$. This is due to the fact that $\textbf{h}_{m}\textbf{h}^{\mathcal{H}}_{m}$ (which produces a rank-one column matrix) and $\textbf{h}^{\mathcal{H}}_{m}\textbf{h}_{m}$ share the same single nonzero eigenvalue, defined as $\lambda$. Note that $\lambda \overset{\text{d}}=\mathcal{Y}\overset{\text{d}}=\mathcal{X}^{2}_{2n}$ \cite{ref30}. Thereby, CDF of $\text{SINDR}_{m}$ is expressed as
\begin{align}
\nonumber
F_{\text{SINDR}_{m}}(\gamma_{\text{th}})&=\text{Pr}\left[\text{SINDR}_{m}\leq \gamma_{\text{th}}\right]\\
&=F_{\mathcal{Y}}\left[\frac{\gamma_{\text{th}}}{\left(\frac{1-(\kappa^{2}_{T}(\omega+1)+\omega)\gamma_{\text{th}}}{(\kappa^{2}_{R}m+N_{0}/p)}\right)}\right],
\label{fsiccdf}
\end{align}
Hence, we can reach (\ref{cdfsindrmmsesicm}), after some straightforward manipulations. Note that $1-(\kappa^{2}_{T}(\omega+1)+\omega)\gamma_{\text{th}}>0$ should hold in (\ref{fsiccdf}) to be a valid CDF.

\subsection{Derivation of (\ref{outasympt})}
\label{appc}
From \cite[Eq. (33)]{ref6}, while assuming that $\frac{p}{N_{0}}\rightarrow \infty$, $F_{p r^{2}_{ii}}(.)$ reads as
\begin{align}
\nonumber
&F_{p r^{2}_{ii}|\frac{p}{N_{0}}\rightarrow \infty}(x)=\\
&\frac{m!(m-i+1)^{1-i}\left(\frac{N_{0}x}{p}\right)^{n-i+1}}{(i-1)!(m-i)!(n-i+1)!}+o\left(\left(\frac{p}{N_{0}}\right)^{-(n-i+1)}\right).
\label{Fprapprox}
\end{align}
Further, based on (\ref{outappen1}), it is obvious that
\begin{align}
\nonumber
&F_{SINDR_{i}|\frac{p}{N_{0}}\rightarrow \infty}(\gamma_{\text{th}})\\
\nonumber
&=\int^{\infty}_{0}\frac{m(m-i+1)^{1-i}\left(\frac{N_{0}\gamma_{\text{th}}((\kappa^{2}_{T}+1)y+\kappa^{2}_{R}m)}{1-\gamma_{\text{th}}\kappa^{2}_{T}}\right)^{n-i+1}}{(i-1)!(m-i)!(n-i+1)!\omega^{m}}\\
&\times y^{m-1}\exp \left(-\frac{y}{\omega}\right)dy+o\left(\left(\frac{p}{N_{0}}\right)^{-(n-i+1)}\right).
\label{Fsindrapprox}
\end{align}
Thus, after performing the binomial expansion \cite[Eq. (1.111)]{ref1} and some straightforward manipulations to (\ref{Fsindrapprox}), (\ref{outasympt}) arises.

\subsection{Derivation of (\ref{cdfsindrmmsesicasym}) and (\ref{cdfsindrmmsesicmasym})}
\label{appf}
Setting $\kappa_{R}=0$ and $\frac{p}{N_{0}}\rightarrow \infty$ in (\ref{cdfsindrmmsesic}), it turns out that only the first summation term significantly impacts the overall outage performance, whereas all other terms approach zero. Thus, setting $k_{1}=1$, $k_{2}=n-m+i+1$ and $j=m-i$, while using the Maclaurin series of the exponential function yields (\ref{cdfsindrmmsesicasym}).

At the last SIC stage, (\ref{cdfsindrmmsesicm}) can alternatively be expressed as
\begin{align*}
P^{(m)}_{\text{out}}(\gamma_{\text{th}})&=\exp\left(-\frac{N_{0}\gamma_{\text{th}}/p}{\left(1-\left(\kappa^{2}_{T}(\omega+1)+\omega\right)\gamma_{\text{th}}\right)}\right)\\
&\times \sum^{\infty}_{k=n}\frac{\left(\frac{N_{0}\gamma_{\text{th}}/p}{\left(1-\left(\kappa^{2}_{T}(\omega+1)+\omega\right)\gamma_{\text{th}}\right)}\right)^{k}}{k!}.
\end{align*}
By retaining only the first summation term (i.e., $k=n$), the final expressions can be extracted.

\ifCLASSOPTIONcaptionsoff
  \newpage
\fi

\end{document}